%% file: main.tex
\documentclass[letterpaper,11pt]{article}

\usepackage[utf8]{inputenc}
\usepackage{graphicx,fullpage,paralist}
\usepackage{comment,hyperref}
\usepackage{thmtools}
\usepackage{tikz}
\usepackage{pgfplots}
\usepackage{varwidth}
\pgfplotsset{compat=1.10}
\usepgfplotslibrary{fillbetween}
\usetikzlibrary{backgrounds}
\usetikzlibrary{patterns}
\usepackage{enumitem}
\usepackage{geometry}
\usepackage{array}
\usepackage{multirow}
\usepackage{stmaryrd}
\usepackage{mathtools}
\usepackage[font=small]{caption}
\usepackage{dsfont}
\usepackage{thm-restate}
\usepackage[ruled]{algorithm2e}

\usepackage{amsmath, amssymb, amsthm}

\newtheorem{theorem}{Theorem}
\newtheorem{lemma}{Lemma}

\newtheorem{definition}{Definition}

\newtheorem{proposition}{Proposition}

\newtheorem{observation}{Observation}

\usepackage{multirow}
\usepackage[noend]{algpseudocode}
\usepackage{xcolor}

\newlength{\algofontsize}
\hypersetup{
	colorlinks,
	linkcolor={red!50!black},
	citecolor={blue!50!black},
	urlcolor={blue!80!black}
}

\newcommand{\vve}[1]{\textcolor{black}{{#1}}}

\newcommand{\pde}[1]{\textcolor{black}{{#1}}}

\newcommand{\feas}{\mathcal{F}}

\newcommand{\calB}{\mathcal{B}}

\newcommand{\calF}{\mathcal{F}}

\newcommand{\EE}{\mathbb{E}}

\newcommand{\Revenue}{\text{Revenue}}
\newcommand{\Welfare}{\text{Welfare}}
\newcommand{\Surplus}{\text{Surplus}}
\newcommand{\AV}{\text{AV}}

\newcommand{\algt}[1]{\ALG_{#1}}

\newcommand{\OPT}{\mathrm{OPT}}
\newcommand{\ALG}{\mathrm{ALG}}

\newcommand{\reals}{\mathbb{R}}

\newcommand{\nil}{\textsf{null}}
\newcommand{\st}{\rho}
\newcommand{\indicator}[1]{\textbf{{1}}[#1]}

\usepackage{algpseudocode}

\title{The Competition Complexity of Prophet Inequalities\thanks{J.~Correa was partially funded by Anillo ICMD (ANID grant ACT210005) and the Center for Mathematical Modeling (ANID grant FB210005). T.~Ezra is supported by the Harvard University Center of Mathematical Sciences and Applications.
M. Feldman was partially funded by the European Research Council (ERC) under the European Union's Horizon 2020 research and innovation program (grant agreement No. 866132), by an Amazon Research Award, by the NSF-BSF (grant number 2020788) and by a grant from TAU Center for AI and Data Science (TAD).  V.~Verdugo was partially funded by Anillo ICMD (ANID grant ACT210005).}}

\author{Johannes Brustle\thanks{London School of Economics, UK. Email: {j.brustle@lse.ac.uk}} \and Jos\'e Correa\thanks{Universidad de Chile, Santiago, Chile. Email: {correa@uchile.cl}} \and Paul D\"utting\thanks{Google Research, Zurich, Switzerland. Email: {duetting@google.com}} \and Tomer Ezra\thanks{Harvard University, Cambridge, USA. Email: {tomer@cmsa.fas.harvard.edu}} \and Michal Feldman\thanks{Tel Aviv University, Israel, and Microsoft ILDC. Email: {mfeldman@tauex.tau.ac.il}} \and Victor Verdugo\thanks{Universidad de O'Higgins, Rancagua, Chile. Email: {victor.verdugo@uoh.ch} }}

\date{}

\begin{document}

\maketitle

\begin{abstract}

We study the classic single-choice prophet inequality problem through a resource augmentation lens. Our goal is to bound the $(1-\varepsilon)$-competition complexity of different types of online algorithms. This metric asks for the smallest $k$ such that the expected value of the online algorithm on $k$ copies of the original instance, is at least a $(1-\varepsilon)$-approximation to the expected offline optimum on a single copy. 

We show that block threshold algorithms, which set one threshold per copy, are optimal and give a tight bound of $k = \Theta(\log \log 1/\varepsilon)$. 
This shows that block threshold algorithms approach the offline optimum doubly-exponentially fast. For single threshold algorithms, we give a tight bound of $k = \Theta(\log 1/\varepsilon)$, establishing an exponential gap between block threshold algorithms and single threshold algorithms.

Our model and results pave the way for exploring resource-augmented prophet inequalities in combinatorial settings. In line with this, we present preliminary findings for bipartite matching with one-sided vertex arrivals, as well as in XOS combinatorial auctions.
Our results have a natural competition complexity interpretation in mechanism design and pricing applications.
\end{abstract}


\input{intro}

\input{model}

\input{dynamic}

\input{static}

\input{extensions}

\bibliographystyle{plain}
\bibliography{biblio}

\appendix
\input{pointmasses}

\input{arrivalorder}

\input{expectedruntime}

\input{combinatorial}

\end{document}

%% file: intro.tex
\section{Introduction}

Studying algorithms beyond the worst case is an increasingly popular approach in theoretical computer science, and is obviously of great practical interest. In this work, we bring together two such directions, which are particularly relevant for online algorithms: The prophet inequality paradigm and resource augmentation. The former restricts the input to be stochastic ---rather than worst-case--- thus enhancing the online algorithm by giving it better information. The latter restricts the optimal algorithm to use fewer resources, thus degrading its performance.

The standard prophet inequality \cite{KrengelS77,KrengelS78} can be interpreted as an allocation problem. There is a single item, which we need to allocate to one of $n$ players that arrive one-by-one in an online fashion. Each player has a value $v_i \ge 0$ for the item. The values $v_1, \ldots, v_n$ are drawn independently   from known distributions $F_1, \ldots, F_n$. We compare the expected value achievable by an online algorithm that has to allocate the item in an online fashion, to the expected offline optimum, which can simply choose the maximum value in the sequence.

The classic result in this area is a tight factor $2$ approximation, which can be achieved with a simple (single) threshold algorithm \cite{SamuelCahn84}. By interpreting the threshold as a price and the players as buyers, this result has immediate implications for several economic and mechanism design applications. It implies that a simple (sequential) posted-price mechanism yields a constant factor approximation to the optimal mechanism \cite{HajiaghayiEtAl07,ChawlaEtAl10}.
Driven in part by this connection to market pricing, the last decade has seen an impressive amount of work on combinatorial extensions of the prophet inequality problem (e.g., for matroids \cite{ChawlaEtAl10,Alaei11,KleinbergW12,FeldmanSZ16}, general downward-closed set systems \cite{RubinsteinS17}, or combinatorial auctions \cite{FeldmanGL15,DuettingFKL17,DuttingKL20,CristiCorrea23,ezra2020pricing}).

Resource augmentation, in turn, has a long  history of success stories in the design and analysis of algorithms. It is particularly popular in scheduling \cite{KalyanasundaramP00,PhillipsSTW97}, but it has also been studied in many other areas such as paging \cite{SleatorT84}, bin packing \cite{CW00}, or selfish routing \cite{RoughgardenT02}. There are typically several plausible notions of resource augmentation, which illuminate different aspects of the problem at hand. In the scheduling literature, for example, 
we grant the online algorithm access to additional machines.
One can also assume it has access to faster machines. In paging, one can analyze the online algorithm with a larger cache relative to that used by the offline algorithm; while in bin packing, one may consider giving the online algorithms larger bins.

In the context of online allocation described above, we grant the online algorithm access to more samples from the same distributions, say the online algorithm can choose to allocate the item to one of $nk$ players, where for each distribution $F_i$, there are $k$ players whose value is drawn from this distribution. The question we seek to answer is that of understanding how many additional resources (samples) the online algorithm needs in order to effectively approximate the performance of the offline optimum in this stochastic setting.

\subsection{The Model}

In this work, we take a resource augmentation approach to prophet inequalities. We study the prophet inequality problem in a setting, where the offline algorithm is handicapped by having less allocation opportunities. This is particularly well motivated in the mechanism design and pricing applications, where it is very likely that the comparatively simpler (sequential) posted-price mechanism attracts additional buyers \pde{compared to the number of buyers that would show up if one was to sell the item through an auction.}

In our model, the online algorithm sees $k$ copies of a prophet inequality instance, where every instance has $n$ values $v_1, \ldots, v_n$ drawn independently from distributions $F_1, \ldots, F_n$. We compare the expected value achieved by the online algorithm on $k$ copies, to the expected maximum value on a single copy. 
Our default assumption is a fixed order model, in which the variables arrive in the same order in each block. 
Our results continue to hold when variables arrive in arbitrary order within each block, and each block may have its own arrival order.

We also consider a further relaxation, the $\gamma$-displacement model, which is parameterized by an integer $\gamma \geq 1$. In this model, the algorithm faces $\gamma k$ copies of the original instance, and an adversary can determine the arrival order, but is constrained by the fact that in each meta-block of $\gamma n$ variables, each type of variable should appear at least once.
This model interpolates between the block model with arbitrary intra-block arrivals 
and the fully adversarial model. 

Similarly to \pde{Eden et al.~\cite{EdenFFTW17a} and Beyhaghi and Weinberg~\cite{BeyhaghiW19}}, for $\varepsilon \ge 0$, we define the \pde{\emph{$(1-\varepsilon)$-competition complexity}} as the smallest $k$ such that the expected value of the online algorithm is guaranteed to be at least a $(1-\varepsilon)$ fraction of the expected maximum value for every instance. 
\pde{This complexity measure has previously been studied by Brustle et al.~\cite{BrustleCDV22} in the context of prophet inequalities and posted pricing, but only for the case of i.i.d.~distributions.}
\pde{They show that while for $\varepsilon = 0$ the competition complexity is unbounded, for $\varepsilon > 0$ it scales as $\Theta(\log \log 1/\varepsilon)$. This shows that the optimal online algorithm (dynamic pricing policy), approaches the optimal offline algorithm (optimal auction) doubly-exponentially fast.}

Of course, 
\pde{the competition complexity} metric is also interesting when restricted to a certain class of online algorithms.
We consider algorithms from three different classes: (1) Single threshold algorithms, which set a single threshold and accept the first value that exceeds the threshold. (2) Block threshold algorithms, which set a threshold for each copy, and accept the first value that exceeds the threshold for its copy. (3) General threshold algorithms, which set a threshold for each of the $nk$ steps and accept the first value that exceeds its threshold. 

Naturally, a simple backwards induction argument shows that threshold algorithms are optimal (in a per instance sense) among all online algorithms. More interestingly, as a first structural insight, we show that with respect to the competition complexity metric, block threshold algorithms are optimal (Proposition~\ref{prop:eq-class} in Section~\ref{sec:model}). We thus generalize the classic result of Samuel-Cahn \cite{SamuelCahn84}, which shows that for $k = 1$ a single threshold algorithm is optimal. Apart from this, threshold algorithms are of particular importance because of their simplicity and natural interpretation as posted-price mechanisms.

\subsection{Our Contribution}

As our main contribution, we resolve the $(1-\varepsilon)$-competition complexity for the classic (single-choice) prophet inequality problem \pde{with non-identical distributions}.
\pde{We show that the competition complexity for general distributions exhibits the same asymptotics as in the i.i.d.~case. Moreover, the optimal asymptotics are attained by block threshold algorithms.}

\medskip
\noindent {\bf Main Result 1 (Theorem~\ref{thm:dynamic}):} For every $\varepsilon > 0$, the $(1-\varepsilon)$-competition complexity of the class of block threshold algorithms is $\Theta(\log \log(1/\varepsilon))$.
\medskip

Our second main result shows a tight bound for single threshold algorithms. 

\medskip
\noindent {\bf Main Result 2 (Theorem~\ref{thm:static}):} For every $\varepsilon > 0$, the $(1-\varepsilon)$-competition complexity of the class of single threshold algorithms is $\Theta(\log(1/\varepsilon))$.
\medskip

\pde{We thus show that also in the case of general, non-identical distributions the best online algorithm (dynamic pricing policy) approaches the best offline algorithm (optimal auction) doubly exponentially fast. Moreover, this is true even when we don't use the full power of dynamic pricing, and prices remain constant within each block.}

In addition, our results show that there is an exponential gap between single threshold algorithms (static pricing policies) and block threshold algorithms (policies that update prices only periodically).

\subsection{Our Techniques}

Our main technical contribution is the upper bound of $O(\log \log(1/\varepsilon))$ on the $(1-\varepsilon)$-competition complexity for block threshold algorithms. 
The matching lower bound of $\Omega(\log \log(1/\varepsilon))$ already applies in the i.i.d.~setting and follows from \cite{BrustleCDV22}.

The main technical ingredient in our proof of the upper bound is an approximate stochastic dominance inequality that allows us to evaluate the performance of any block threshold algorithm, with decreasing thresholds $\tau_1>\tau_2>\ldots,\tau_k$ (Lemma \ref{lem:dynamic-1}). The approximation factor of this stochastic dominance inequality is parameterized by $\tau_1,\ldots,\tau_k$, and it is obtained  
\pde{through a careful} analysis of the minimum stochastic-dominance approximation factor achievable on each of the $k+1$ sub-intervals $[0,\tau_k),[\tau_k,\tau_{k-1}),\ldots,[\tau_1,\infty)$ defined by the block thresholds. 
Then, to find the best possible approximation guarantee for block threshold algorithms that can be obtained in this way, we need to solve a (high-degree polynomial) max-min optimization problem. Rather than solving this problem exactly, we provide a tight double-exponentially fast increasing lower bound on the value of this max-min problem by constructing an explicit set of thresholds (Lemma~\ref{lem:dynamic-2}). \pde{In combination with the lower bound, this} shows that our stochastic dominance approach provides an optimal competition complexity bound up to a constant factor.

Our upper bound of $O(\log(1/\varepsilon))$ for single threshold algorithms follows in a rather direct way from the ``median rule'' proof of \cite{SamuelCahn84}. 
Our key insight for this case is that the upper bound is asymptotically tight, which we show by providing an explicit lower bound construction.

\subsection{Other Arrival Orders}

We also explore the robustness of the $(1-\varepsilon)$-competition complexity metric to different assumptions about the arrival order.  
To study this effect, we introduce the $\gamma$-displacement model, where $\gamma \geq 1$ is a parameter. 
In this model, the algorithm faces $\gamma k$ copies of the original instance. The arrival order is determined by an adversary, but the adversary is constrained by the requirement that, within each meta-block of $\gamma n$ variables, each type of variable should appear at least once. 
While for $\gamma = 1$ the adversary is restricted to move variables within their block, for $\gamma > 1$ the adversary can also move variables across blocks.

For this model, we show that for every $\gamma \geq 1$, the $(1-\varepsilon)$-competition complexity of block threshold algorithms is $O(\gamma \log \log (1/\varepsilon))$ (Proposition~\ref{prop:displace}). Comparing this with the $(1-\varepsilon)$-competition complexity of block threshold algorithms in the default model, this shows that the competition complexity is increased by a multiplicative factor of $\gamma$, but the scaling behavior in $1/\varepsilon$ remains the same. This shows that the competition complexity of block threshold algorithms degrades gracefully as we move away from the default model.

We also show that a comparable result cannot be achieved for single threshold algorithms.
Namely, for this class of algorithms, we show that there exists some $\gamma > 1$, such that the $(1-\varepsilon)$-competition complexity of single threshold algorithms is least $\Omega(1/\varepsilon^{1/3})$ (Proposition~\ref{prop:displace-single-threshold}). This shows that the $(1-\varepsilon)$-competition complexity of this type of algorithms transitions from growing logarithmically in $1/\varepsilon$ when $\gamma = 1$ to growing (at least) polynomially in $1/\varepsilon$ when $\gamma > 1$. 

We complement these results, with a lower bound on the $(1-\varepsilon)$-competition complexity of general threshold algorithms in the fully adversarial model, showing that the $(1-\varepsilon)$-competition complexity is at least $\Omega(1/\varepsilon)$ (Proposition~\ref{prop:adversarial}).

\subsection{Extensions}

Our work opens up the question of studying resource-augmented prophet inequalities for richer combinatorial settings. We present some preliminary results for submodular and XOS combinatorial auctions. The proper generalization of threshold algorithms for this setting are prices, and similar to the distinction between block threshold algorithms and single threshold algorithms, we can distinguish between block-consistent prices which stay fixed within a block and static prices that remain fixed throughout.

\medskip
\noindent{\bf Additional Result 1 (Theorem~\ref{thm:xos-dynamic}):}
The $(1-\varepsilon)$-competition complexity of block-consistent prices for submodular and XOS combinatorial auctions is $O(\log(1/\varepsilon))$. 
\medskip

\medskip
\noindent{\bf Additional Result 2 (Theorem~\ref{thm:xos-static}):}
The $(1-\varepsilon)$-competition complexity of static prices for submodular and XOS combinatorial auctions is $O(1/\varepsilon)$. 
\medskip

Our results present a first glimpse at a potentially rich theory and already show that the constant-factor that can be shown in the single-shot setting  \cite{FeldmanGL15,DuettingFKL17} vanishes exponentially fast with block-consistent prices with additional resources. Whether this can also be achieved with static prices remains open, just as the question of whether a double-exponentially fast approach is possible with dynamic prices. 
More generally, it would be interesting to establish a formal separation between static and dynamic prices. 
It also remains open whether comparable results can be obtained for subadditive combinatorial auctions \cite{DuttingKL20,CristiCorrea23}.

Of course, the same question can be studied for other combinatorial settings such as matroids \cite{ChawlaEtAl10,KleinbergW12} or matching constraints \cite{GravinW19,ezra2022prophet}. Finally, we hope that our paper will spark work on other notions of resource augmentation.

\subsection{Related Work}

A seminal paper in Economic Theory that is very relevant to our work is a paper by Bulow and Klemperer \cite{BulowKlemperer96}. This paper argues that in an i.i.d.~auction setting rather than going with the complex revenue-maximizing auction, one can attract one additional buyer and simply run a second price auction. The revenue that results from the simpler auction with one additional buyer will be higher than that from the optimal auction.

This work obviously found its audience in computer science, where it inspired work on the competition complexity of simple mechanisms, see for example the excellent work of Eden et al.~\cite{EdenFFTW17a}, Beyhaghi and Weinberg \cite{BeyhaghiW19}, and Babaioff et al.~\cite{BabaioffGG20}.

Our work is also related to the recent work of \cite{AbolhassaniEEHK17,LiuLPSS21} through the notion of ``frequent instances''. Both of these papers study the \emph{frequent prophets} problem, in which each distribution must be repeated at least some number of times: An instance is $m$-frequent if each distribution appears at least $m$ times. The inspiration behind frequent prophets comes from the difficulty of directly analyzing the general independent case in random and free-order models, in which the values either arrive in uniform random order or the algorithm is free to choose the order of observation, and the idea is to bring the instance closer to the i.i.d.~case. 
Abolhassani et al.~\cite{AbolhassaniEEHK17} show how to design a $0.738$-competitive policy for $O(1)$-frequent instances in the free-order model and $\Theta(\log n)$-frequent instances in the random-order model, while Liu et al.~\cite{LiuLPSS21}  improve on the random-order results by showing how to obtain a $(\beta-\epsilon)$-competitive policy for $O_\varepsilon(1)$-frequent instances, where $\beta\approx 0.745$ is the optimal guarantee for the i.i.d.~prophet inequality problem \cite{CorreaFHOV21}. The difference to our work is that they study both ALG and OPT on ``frequent instances'', while we compare the algorithm in a $k$-frequent instance with the prophet on a standard ($1$-frequent) instance.
Another difference is that we work in the fixed order setting, whereas the results in \cite{AbolhassaniEEHK17,LiuLPSS21} apply to the random and free-order model.

\subsection{Organization}

The paper is organized as follows. We formally state our model and the competition complexity metric in Section~\ref{sec:model}. We also introduce single threshold algorithms and block threshold algorithms, and establish the optimality of the latter. In Section~\ref{sec:block} we present our results for block threshold algorithms. Afterwards, in Section~\ref{sec:static}, we turn our attention to single threshold algorithms. We discuss  combinatorial extensions in  Section~\ref{sec:extensions}. We defer the discussion of additional arrival orders 
to Appendix~\ref{sec:arrivalorder}.

%% file: model.tex
\section{Preliminaries}\label{sec:model}

\paragraph{The Prophet Inequality Problem.}

Consider the following game between an online algorithm (``gambler'') and an offline algorithm (``the prophet'').  The online algorithm gets to observe a sequence of $n$ (non-negative) numbers $v_1, \ldots, v_n$ one-by-one. Each $v_i$ is drawn independently from a known distribution $F_i \in \Delta$, where $\Delta$ is the set of all distributions over $\reals_{\geq 0}$. 
We call a sequence of distributions $F_1, \ldots, F_n$ an instance. 

We refer to the online algorithm as ALG. Upon seeing a value $v_i$ the online algorithm has to immediately and irrevocably decide whether to accept the current value $v_i$ and stop the game, or to proceed to the next value $v_{i+1}$.  Every online algorithm induces a stopping time 
$\st \in [n] \cup \{\nil\}$, where $\st$ is the index of the value chosen by the online algorithm (to handle the case where the online algorithm does not accept any value, we set $\st = \nil$ and interpret $v_{\nil}$ as zero). The algorithm's expected reward is $\EE[\ALG(v)] = \EE[v_{\st}]$. The offline algorithm, in contrast, can see the entire sequence of values $v_1, \ldots, v_n$ at once and can simply choose the maximum value $\max_{i \in [n]} v_i$. The offline algorithm's expected reward is $\EE[\max_{i\in[n]}v_i]$.

In the Prophet Inequality problem the online algorithm is evaluated by its \emph{competitive ratio}, defined as the worst-possible ratio (over all instances)  between the online algorithm's expected reward $\EE[\ALG(v)] = \EE[v_{\st}]$ and the offline algorithm's expected reward   $\EE[\max_{i \in [n]} v_i]$. Let $\alpha \in [0,1]$. An online algorithm is $\alpha$-competitive if
\[
\inf_{F_1, \ldots, F_n\in \Delta} \frac{\EE_{v \sim (F_1,\ldots,F_n)}[\ALG(v)]}{\EE_{v \sim (F_1,\ldots,F_n)}[\max_{i \in [n]} v_i]} \geq \alpha. 
\]

\paragraph{The Competition Complexity Benchmark.} Our goal is to compare the expected performance of an online algorithm on $k$ independent copies of the original instance, to the expected value of the offline algorithm on a single instance. This can be done under different assumptions on how the $nk$ values of the $k$ copies are presented to the online algorithm. 

Our default model is what we call the \emph{block model}. In this model, the online algorithm sees the $k$ copies of the original instance one after the other. We refer to each copy as a block. Within each block, the $n$ values arrive in the same order as in the original instance.

More formally, when $(F_1, \ldots, F_n)$ is the original instance, we denote with $(F_1,\ldots,F_n)^k$ the instance with $k$ copies. 
The input to the online algorithm consists of $kn$ numbers 
$$v_1^{(1)}, \ldots, v_n^{(1)},v_1^{(2)}, \ldots, v_n^{(2)},\ldots,v_1^{(k)}, \ldots, v_n^{(k)},$$ where each $v_i^{(j)}$ for $i \in [n]$ and $j \in [k]$ is an independent draw from $F_i$. The offline algorithm receives only $n$ numbers $v_1, \ldots, v_n$ where each $v_i$ for $i \in [n]$ is an independent draw from $F_i$.

Let $\mathcal{A}$ be a family of online algorithms. Let $\mathcal{A}_{n,k}$ be defined as all the online algorithms in $\mathcal{A}$ that are defined on $\Delta_{n,k} = \{(F_1,\ldots,F_n)^k \mid F_1, \ldots, F_n \in \Delta\}$.

\begin{definition}[Competition complexity]\label{def:competition} Given $\varepsilon \ge 0$, the $(1-\varepsilon)$-competition complexity with respect to a class of algorithms $\mathcal{A}$ is the smallest positive integer number $k(\varepsilon)$ such that for every $n$, every $F_1, \ldots, F_n\in \Delta$, and every $k\ge k(\varepsilon)$, it holds that
\[
\max_{\ALG \in \mathcal{A}_{n,k}}\EE_{v \sim (F_1, \dots, F_n)^k}[\ALG(v)] \geq (1-\varepsilon) \cdot \EE_{v \sim (F_1, \dots, F_n)}[\max_{i \in [n]} v_i].
\]
\end{definition}

The case $\varepsilon = 0$ is also referred to as \emph{exact} competition complexity; it was shown by Brustle et al.~\cite[Theorem 2.1]{BrustleCDV22} that the exact competition complexity is unbounded even for the i.i.d.~case (namely, where $F_i=F_j$ for all $i,j$). So, naturally, our focus will be on the case $\varepsilon > 0$.

We remark that, as is common in the literature, we focus on a worst-case notion of competition complexity, which asks for the minimum number of copies that suffices for a worst-case approximation guarantee. In Appendix~\ref{sec:expectedruntime}, we explore an alternative, which asks for the expected number of copies that are required to achieve this.

\paragraph{Classes of Online Algorithms.} 
We are particularly interested in  three types of online algorithms. In order of increasing generality these are:
\begin{itemize}
\item A \emph{single threshold algorithm} is defined by a single threshold $\tau$ and it accepts the first value $v_\ell$ (indexed according to arrival order) that is at least $\tau$. 
\item A \emph{block threshold algorithm} sets $k$ thresholds $\boldsymbol\tau = (\tau_1, \ldots, \tau_k)$, i.e., one threshold per block. It accepts the first value $v_i^{(j)}$ that is larger than the threshold $\tau_j$ for its block.
\item A \emph{general threshold algorithm} sets $nk$ thresholds $\boldsymbol\tau = (\tau_1, \ldots, \tau_{n k})$, and accepts the first value $v_\ell$ (again, indexed according to arrival order) that is at least $\tau_{\ell}$.
\end{itemize}

In all of the above cases, we allow the algorithm to accept a value only with a certain probability in case it exactly meets the threshold. This ability to randomize is relevant only for distributions with point masses.
A standard backward-induction argument shows that the optimal online algorithm is a general threshold algorithm.

In the classic prophet inequality problem, a single threshold algorithm attains the best possible competitive ratio \cite{SamuelCahn84}. Proposition~\ref{prop:eq-class} generalizes this result, and shows that for the more general competition complexity benchmark, it is without loss to focus on block threshold algorithms.
\pde{The main idea behind this reduction} \vve{is that finding the competition complexity of general threshold algorithms is, essentially, equivalent to understanding the competition complexity for instances where every distribution is a weighted Bernoulli. For these instances, and within each block, only non-zero values of a suffix of the block are chosen, and therefore it is sufficient to implement a single threshold per block.
This implies that block threshold algorithms are as powerful as general threshold algorithms for the block model.}

\begin{restatable}{proposition}{propeqclass}
\label{prop:eq-class}
For every $\varepsilon \in (0,1) $, the $(1-\varepsilon)$-competition complexity with respect to the class of block threshold algorithms is the same as  the $(1-\varepsilon)$-competition complexity with respect to the class of general threshold algorithms.
\end{restatable}

\begin{proof}
Let $k'$ be the $(1-\varepsilon)$-competition complexity of the class of block threshold algorithms, and let $k$ be the $(1-\varepsilon)$-competition complexity of the class of general threshold algorithms.
Clearly $k' \geq k$, since every block threshold algorithm is a general threshold algorithm.
We next prove the other inequality.
By the definition of $k'$, there exists an instance  $(F_1,\ldots,F_n)$, such that for every block threshold algorithm $\ALG$, it holds that 
\begin{equation}
\label{eq:1-var}    
\EE_{v \sim (F_1, \dots, F_n)^{k'-1}}[\ALG(v)] < (1-\varepsilon) \cdot \EE_{v \sim (F_1, \dots, F_n)}[\max_{i \in [n]} v_i].
\end{equation}
Given a block threshold algorithm $\ALG$, let $$\alpha= \frac{1}{4}\Big((1-\varepsilon) \cdot \EE_{v \sim (F_1, \dots, F_n)}[\max_{i \in [n]} v_i] - \EE_{v \sim (F_1, \dots, F_n)^{k'-1}}[\ALG(v)]\Big),$$ 
and let 
$$\beta^* = \inf \Big\{\beta\ge 0 :  \EE_{v \sim (F_1, \dots, F_n)}[\max_{i \in [n]} v_i \cdot \indicator{\max_{i \in [n]} v_i \geq \beta} ]  \leq \alpha \Big\}.$$
Let $m = \lceil \beta^*/\alpha \rceil$, and for every $i\in [m]$, let $p_i = \Pr[(i-1)\cdot \alpha \leq  \max_{i \in [n]} v_i<  i \cdot \alpha]$.

For every $i\in \{1,\ldots,m-1\}$, let $D_i$ be the weighted Bernoulli distribution that takes the value $i\cdot \alpha $ with probability $p_{i+1}/(1-\sum_{j=i+2}^m p_j)$, and zero otherwise. 
Let  $M_1 = \max_{i\in [m-1]}w_i$ with $w \sim (D_1,\ldots,D_{m-1})$, and  $M_2 = \max_{i\in [n]}v_i$ with $v \sim (F_1,\ldots,F_{n})$. It holds that $M_1$ has the same distribution as $\lfloor M_2 \cdot \indicator{M_2 <m \cdot \alpha}/\alpha  \rfloor \cdot \alpha $, since for every $r\in [m-1]$ we have

\begin{align}
\Pr[M_1= r\cdot \alpha ] & =  \frac{p_{r+1}}{1-\sum_{j=r+2}^m p_j}  \cdot \prod_{r'=r+1}^{m-1} \left(1- \frac{p_{r'+1}}{1-\sum_{j=r'+2}^m p_j}\right) \nonumber\\  
&=   \frac{p_{r+1}}{1-\sum_{j=r+2}^m p_j}  \cdot \prod_{r'=r+1}^{m-1} \left( \frac{1-\sum_{j=r'+1}^m p_{j}}{1-\sum_{j=r'+2}^m p_j}\right)\nonumber\\ 
&= p_{r+1}=\Pr\left[\left\lfloor \frac{M_2 \cdot \indicator{M_2 <m \cdot \alpha}}{\alpha } \right\rfloor \cdot \alpha=r\cdot \alpha \right] . \label{eq:probs}
\end{align}
Thus, for the instance $(D_1,\ldots,D_{m-1})$, it holds that 
\begin{equation}
    \EE_{w \sim (D_1, \dots, D_{m-1})}[\textstyle\max_{i \in [m-1]} w_i] \geq \EE_{v \sim (F_1, \dots, F_n)}[\max_{i \in [n]} v_i] -2 \alpha,\label{eq:alpha}
\end{equation}
since an $\alpha$ term is lost due to values above $m\cdot \alpha \geq \beta^*$, and another $\alpha $ term is lost due to the flooring.
On the other hand, since the thresholds calculated by the best general threshold algorithm for $(D_1,\ldots,D_{m-1})^{k'-1}$ are monotonically decreasing, and since the distributions of each block are weighted Bernoulli variables with increasing weights, it holds that within each block, only non-zero values of a suffix of the block are chosen.
Thus, the optimal algorithm $\ALG^{\star}$ for instance $(D_1,\ldots,D_{m-1})^{k'-1}$ is a block threshold algorithm.

Let $\tau_1^{\star}\geq \ldots\geq \tau_{k'-1}^{\star}$ be the monotone decreasing thresholds used by $\ALG^{\star}$. 
Consider the algorithm $\ALG'$ that selects a value in block $\ell$  if it is in the interval $[\alpha \cdot (1+\lfloor \tau^{\star}_\ell/\alpha\rfloor), m\cdot \alpha)$. By Equation~\eqref{eq:probs}, the probability that $\ALG'$ selects an element in every iteration is the same as  $\ALG^{\star}$, and given that both algorithms stop at block $\ell$, the expectation of $\ALG'$ is at least  the expectation of $\ALG^{\star}$.
Thus, $$\EE_{w \sim (D_1, \dots, D_{m-1})^{k'-1}}[\ALG^{\star}(w)] \leq \EE_{v \sim (F_1, \dots, F_n)^{k'-1}}[\ALG'(v)] \leq \EE_{v \sim (F_1, \dots, F_n)^{k'-1}}[\ALG(v)],$$ where the second inequality is by noticing that if we select values above $m\cdot\alpha$, it can only improve the performance, and the algorithm becomes a block threshold algorithm.
Therefore, 
\begin{align*}
& (1-\varepsilon)\cdot \EE_{w\sim (D_1, \dots, D_{m-1})}[\textstyle\max_{i \in [m-1]} w_i] - \EE_{w \sim (D_1, \dots, D_{m-1})^{k'-1}}[\ALG^{\star}(w)] \\ 
&\ge (1-\varepsilon)\cdot \EE_{w \sim (D_1, \dots, D_{m-1})}[\textstyle\max_{i \in [m-1]} w_i] - \EE_{v \sim (F_1, \dots, F_n)^{k'-1}}[\ALG(v)] \\
& = (1-\varepsilon)\cdot \EE_{w \sim (D_1, \dots, D_{m-1})}[\textstyle\max_{i \in [m-1]} w_i] -  (1-\varepsilon) \cdot \EE_{v \sim (F_1, \dots, F_n)}[\max_{i \in [n]} v_i] + 4\alpha \\ 
& \geq  (1-\varepsilon)  (-2\alpha) +4\alpha > 0,  
\end{align*}
where the equality is by definition of $\alpha$, the first inequality is by Equation~\eqref{eq:alpha}, and the last inequality is since $\varepsilon<1$, and since by Equation~\eqref{eq:1-var}, $\alpha >0$.
Thus, $k> k'-1$, which concludes the proof.
\end{proof}

\paragraph{Beyond the Block Model.} Our default model is the block model, according to which we repeat an instance $k$ times, keeping the arrival order within each block the same. Our results continue to hold, even if the arrival order within each block is arbitrary. We discuss this and additional models along with implications for our results 
in Appendix~\ref{sec:arrivalorder}.

\medskip

As is standard in the literature, to simplify the presentation, we generally assume that the distributions do not admit point masses (i.e., $F_i$ is continuous for every $i\in [n]$). In Appendix~\ref{app:pointmass}, we discuss how to adjust our algorithms for cases with point masses.

The no point masses assumption simplifies the exposition because in the absence of point masses every quantile is associated with a threshold. With point masses, this is not necessarily the case, but the problem can be resolved using randomization.

%% file: dynamic.tex
\section{Block Threshold Algorithms}\label{sec:block}

In this section, we study the competition complexity of the class of block threshold algorithms. 
The following is the main result of this section.

\begin{theorem}\label{thm:dynamic}
For every $\varepsilon>0$ the
 $(1-\varepsilon)$-competition complexity of the class of block threshold algorithms is $\Theta\left(\log\log(1/\varepsilon)\right)$.
\end{theorem}

In what follows, given $\boldsymbol\tau=(\tau_1,\ldots,\tau_k)$ we denote by $\algt{{\boldsymbol \tau}}$ the block threshold algorithm such that for every copy $j\in \{1,\ldots,k\}$ it selects the first element that exceeds $\tau_j$, if such element exists. 
We denote $\tau_0=\infty$ and $\tau_{k+1}=0$. 

The remainder of this section is organized as follows. 
As a warm-up, we start by resolving the case of $k=2$, which motivates and presents the main ideas of our approach. 
In Section~\ref{sec:reduction} we 
establish a lemma that allows us to evaluate the performance of any block threshold algorithm $\algt{{\boldsymbol \tau}}$ by establishing an approximate stochastic dominance inequality that extends the case of $k=2$ to every $k$.
This inequality is parameterized by the thresholds $\tau= (\tau_1, \ldots, \tau_k)$. 
In order to get the best possible approximation factor, we need to solve the induced max-min problem.
In Section~\ref{sec:feasible} we solve the corresponding max-min problem by presenting an explicit feasible solution.
Finally, in Section~\ref{sec:full-proof}, we present the full proof of Theorem~\ref{thm:dynamic}, based on the ingredients established in previous sections.\\

\noindent{{\it Warm-up: The case of $k=2$}.} To design a block threshold algorithm for $k=2$, we compute the thresholds $\tau_1,\tau_2$ by finding the appropriate quantiles of the distribution of $\max_{i\in [n]}v_i$, with $v_i\sim F_i$ for every $i\in \{1,\ldots,n\}$.
More specifically, for $j\in \{1,2\}$, let $p_j=\Pr_{v \sim (F_1, \dots, F_n)}[\max_{i\in [n]}v_i \geq \tau_j]$, where $\tau_1$ and $\tau_2$ will be determined by specifying the values of $p_1$ and $p_2$.
By our assumption of $F_1,\ldots,F_n$ having no point masses, any $p_1,p_2 \in [0,1]$ corresponds to a pair of thresholds $\tau_1$ and $\tau_2$. 
To establish the competition complexity result involving the expectations of $\algt{\boldsymbol{\tau}}$ and $\max_{i\in [n]}v_i$, our goal is to state an approximate stochastic dominance result between these two random variables.
By setting $\tau_1\ge \tau_2$ and such that $p_1>0$, we can show that 
\begin{align}
    \Pr_{v \sim (F_1, \dots, F_n)^2}[\algt{{\boldsymbol \tau}} (v) \geq x] &\ge \phi_1(p_1,p_2)\Pr_{v \sim (F_1, \dots, F_n)}[\max_{i \in [n]} v_i \geq x] \quad \text{ when }x\ge \tau_1,\nonumber\\
    \Pr_{v \sim (F_1, \dots, F_n)^2}[\algt{{\boldsymbol \tau}} (v) \geq x] &\ge \phi_2(p_1,p_2)\Pr_{v \sim (F_1, \dots, F_n)}[\max_{i \in [n]} v_i \geq x] \quad \text{ when }\tau_1>x\ge \tau_2,\nonumber\\
    \Pr_{v \sim (F_1, \dots, F_n)^2}[\algt{{\boldsymbol \tau}} (v) \geq x] &\ge \phi_3(p_1,p_2)\Pr_{v \sim (F_1, \dots, F_n)}[\max_{i \in [n]} v_i \geq x] \quad \text{ when }\tau_2>x,\label{max-min-2}
\end{align}
where $\phi_1(p_1,p_2)=1-p_1+(1-p_1)(1-p_2)$, $\phi_2(p_1,p_2)=p_1/p_2+(1-p_1)(1-p_2)$, and $\phi_3(p_1,p_2)=p_1+p_2(1-p_1)$.
In particular, we get that
\begin{equation*}
\Pr_{v \sim (F_1, \dots, F_n)^2}[\algt{{\boldsymbol \tau}} (v) \geq x] \ge \min\Big\{\phi_1(p_1,p_2),\phi_2(p_1,p_2),\phi_3(p_1,p_2)\Big\}\Pr_{v \sim (F_1, \dots, F_n)}[\max_{i \in [n]} v_i \geq x]
\end{equation*}
for every $x\ge 0$.
The above inequality is stated in its general form for every $k$ in Lemma \ref{lem:dynamic-1}.
Then, in order to get the best possible approximation factor in this way, we solve the following max-min optimization problem: 
$$\max\Big\{\min\Big\{\phi_1(p_1,p_2),\phi_2(p_1,p_2),\phi_3(p_1,p_2)\Big\}:0<p_1\le p_2,\; p_1,p_2\in [0,1]\Big\}.$$
It can be shown that the optimal solution for this problem is attained at $p_1=2/5$ and $p_2 = 2/3$, which yields a factor of $4/5$.
For the case of general $k$, we do not solve exactly this max-min problem as it is a high-dimensional polynomial optimization problem,
but we construct explicitly a feasible solution that yields the optimal competition complexity guarantee up to a constant factor.
This is formalized in Lemma \ref{lem:dynamic-2}.

\medskip

\subsection{Reduction to Max-Min Problem via Approximate Stochastic Dominance}
\label{sec:reduction}
We start by showing a lemma that allows us to evaluate the performance of any block threshold algorithm $\algt{{\boldsymbol \tau}}$ by establishing an approximate stochastic dominance inequality that extends \eqref{max-min-2} to every $k$.
When $k=2$, the inequalities in \eqref{max-min-2} are obtained by splitting the domain of $x$ according to the three different sub-intervals defined by the two thresholds.
We exploit this idea to get an approximate stochastic dominance inequality by studying each of the $k+1$ ranges defined by the $k$ block thresholds.

More specifically, given $F_1,\ldots,F_n\in \Delta$, and $ \boldsymbol\tau=(\tau_1,\ldots,\tau_k)$ such that $p_0=0$ and $p_{\ell} =\Pr_{v \sim (F_1, \dots, F_n)}[\max_{j\in [n]}v_j \geq \tau_{\ell}]>0$ for every $\ell\in \{1,2,\ldots,k\}$, let
\begin{align*}
    \Phi_1(F_1,\ldots,F_n,\boldsymbol\tau)&=\sum_{\ell=1}^{k}\prod_{j=1}^{\ell}(1-p_j),\\
    \Phi_i(F_1,\ldots,F_n,\boldsymbol\tau)&=\frac{1}{p_i}\sum_{\ell=1}^{i-1}p_{\ell}\prod_{j=0}^{\ell-1}(1-p_j) + \sum_{\ell=i}^{k}\prod_{j=0}^{\ell}(1-p_j) \quad \text{for every $i\in \{2,\ldots,k\}$}, \text{ and}\\
    \Phi_{k+1}(F_1,\ldots,F_n,\boldsymbol\tau)&=\sum_{\ell=1}^{k}p_{\ell}\prod_{j=0}^{\ell-1}(1-p_j).
\end{align*} 
Let $\Phi(F_1,\ldots,F_n,\boldsymbol\tau) =\min_{i\in \{1,\ldots,k+1\}}\Phi_i(F_1,\ldots,F_n,\boldsymbol\tau).$
We say that $ \boldsymbol\tau=(\tau_1,\ldots,\tau_k)$ is {\it decreasing} if $\tau_j > \tau_{j+1}$ for every $j\in \{1,\ldots,k\}$.

\begin{lemma}\label{lem:dynamic-1}
For every $F_1,\ldots,F_n\in \Delta$, every $x\ge 0$, and every decreasing $\boldsymbol\tau=(\tau_1,\ldots,\tau_k)$ such that $\Pr_{v \sim (F_1, \dots, F_n)}[\max_{j\in [n]}v_j \geq \tau_1]>0$, we have
\begin{equation*}
\Pr_{v \sim (F_1, \dots, F_n)^k}[\algt{{\boldsymbol \tau}} (v)\geq x] \ge \Phi(F_1,\ldots, F_n,\boldsymbol\tau)\Pr_{v \sim (F_1, \dots, F_n)}[\max_{i \in [n]} v_i \geq x]. 
\end{equation*}
\end{lemma}

\begin{proof}
Given $x\ge 0$, let 
$\ell(\boldsymbol\tau,x)=\min\{\ell\in \{1,\ldots,k+1\}:x\ge \tau_{\ell}\}$, that is, $\ell(\boldsymbol\tau,x)$ is the first block $\ell$ for which $x$ is at least the threshold $\tau_{\ell}$.
Let $\calB$ be defined as follows: If there exists $\ell\in \{1,\ldots,k\}$ such that $\max_{i\in [n]}v_i^{(\ell)}>\tau_{\ell}$, then $$\mathcal{B}=\min\Big\{\ell\in \{1,\ldots,k\}:\max_{i\in [n]}v_i^{(\ell)}>\tau_{\ell}\Big\},$$
and $\calB=k+1$ otherwise.
That is, $\calB$ is equal to the first block for which there exists a value in the block that surpasses the block threshold, and it is equal $k+1$ in case such value does not exist.
Recall that we denote $p_0=0$ and $p_{\ell} =\Pr_{v \sim (F_1, \dots, F_n)}[\max_{j\in [n]}v_j \geq \tau_{\ell}]$ for every $\ell\in \{1,2,\ldots,k+1\}$; in particular $p_{k+1}=1$.
The following holds:

\begin{align}
    &\Pr_{v \sim (F_1, \dots, F_n)^k}[\algt{{\boldsymbol \tau}} (v)\geq x]\nonumber\\
    &=\sum_{\ell=1}^{k}\Pr_{v \sim (F_1, \dots, F_n)^k}[\algt{{\boldsymbol \tau}} (v)\geq x \text{ and } \calB = \ell]\nonumber\\
    &=\sum_{\ell=1}^{k}\Pr_{v \sim (F_1, \dots, F_n)^k}[\algt{{\boldsymbol \tau}} (v)\geq x \text{ and } \max_{i \in [n]}v_i^{(\ell)} \geq \tau_{\ell} \mid \calB> \ell-1]\Pr_{v \sim (F_1, \dots, F_n)^k}[\calB> \ell-1]\nonumber\\
    &\geq\sum_{\ell=1}^{k}\Pr_{v \sim (F_1, \dots, F_n)^k}[v_j^{(\ell)} \notin [\tau_{\ell},x)  \text{ for all } j\in [n] \text{ and } \max_{i \in [n]}v_i^{(\ell)} \geq \tau_{\ell} \mid \calB> \ell-1] \prod_{j=0}^{\ell-1}(1-p_j)\nonumber\\
    &=\sum_{\ell=1}^{k}\Pr_{v^{(\ell)} \sim (F_1, \dots, F_n)}[v_j^{(\ell)} \notin [\tau_{\ell},x)  \text{ for all } j\in [n] \text{ and } \max_{i \in [n]}v_i^{(\ell)} \geq \tau_{\ell}] \prod_{j=0}^{\ell-1}(1-p_j),\label{lem1-summation}
\end{align}
where in the second equality, we use that $\calB=\ell$ is equivalent to the algorithm not stopping before reaching the block $\ell$, and in block $\ell$ having $\max_{i \in [n]}v_i^{(\ell)} \geq \tau_{\ell}$. 
For the inequality, we observe that conditioned on $\calB> \ell-1$, the intersection of the event $\algt{{\boldsymbol \tau}} (v)\geq x$ with $\max_{i \in [n]}v_i^{(\ell)} \geq \tau_{\ell}$ contains $v_j^{(\ell)} \notin [\tau_{\ell},x)  \text{ for all } j\in [n]$. The final equality follows from independence across copies.

Suppose that $\ell(\boldsymbol{\tau},x)\ge 2$ and let $\ell\in \{1,\ldots, \ell(\boldsymbol{\tau},x)-1\}$.
In particular, we have $\tau_{\ell}>x$,
and therefore,
\begin{align}
    &\Pr_{v \sim (F_1, \dots, F_n)^k}[v_j^{(\ell)} \notin [\tau_{\ell},x)  \text{ for all } j\in [n] \text{ and } \max_{i \in [n]}v_i^{(\ell)} \geq \tau_{\ell}]\nonumber\\
    &= \Pr_{v^{(\ell)} \sim (F_1, \dots, F_n)}[\max_{i\in [n] }v_i^{(\ell)} \geq \tau_\ell]\nonumber\\
    &= \Pr_{v \sim (F_1, \dots, F_n)}[\max_{i\in [n] }v_i \geq x]\cdot \frac{\Pr_{v \sim (F_1, \dots, F_n)}[\max_{i\in [n]}v_i \geq \tau_\ell]}{\Pr_{v \sim (F_1, \dots, F_n)}[\max_{i\in [n] }v_i \geq x]}\nonumber \\
    &\ge \Pr_{v \sim (F_1, \dots, F_n)}[\max_{i\in [n] }v_i \geq x]\cdot \frac{\Pr_{v \sim (F_1, \dots, F_n)}[\max_{i\in [n]}v_i \geq \tau_\ell]}{\Pr_{v \sim (F_1, \dots, F_n)}[\max_{i\in [n] }v_i \geq \tau_{\ell(\boldsymbol\tau,x)}]} \nonumber\\
    &= \Pr_{v \sim (F_1, \dots, F_n)}[\max_{i\in [n] }v_i \geq x]\cdot\frac{p_\ell}{p_{\ell(\boldsymbol\tau,x)}},\label{lem1-case1}
\end{align}
where the  inequality holds since $x\geq \tau_{\ell(\boldsymbol\tau,x)}$.

Now suppose that $\ell(\boldsymbol{\tau},x)\le k$ and let $\ell\in \{\ell(\boldsymbol{\tau},x),\ldots,k\}$.
In particular, we have $x\geq\tau_{\ell}$, and therefore, 
\begin{align}
    &\Pr_{v^{(\ell)} \sim (F_1, \dots, F_n)}[v_j^{(\ell)} \notin [\tau_{\ell},x)  \text{ for all } j\in [n] \text{ and } \max_{i \in [n]}v_i^{(\ell)} \geq \tau_{\ell}]\nonumber\\
    &\ge \Pr_{v^{(\ell)} \sim (F_1, \dots, F_n)}[v_j^{(\ell)} \notin [\tau_{\ell},x)  \text{ for all } j\in [n]\text{ and }\max_{i \in [n]}v_i^{(\ell)} \geq x]\nonumber\\
    &\ge  \sum_{i=1}^n\Pr_{v \sim (F_1, \dots, F_n)}[v_i \geq x\text{ and } v_j < \tau_{\ell} \text{ for all } j \in [n]\setminus\{i\}]\nonumber\\
    &=\sum_{i=1}^n\Pr_{v \sim (F_1, \dots, F_n)}[v_i \geq x]\Pr_{v \sim (F_1, \dots, F_n)}[v_j < \tau_{\ell} \text{ for all } j \in [n] \setminus\{i\} ]\nonumber\\
    &\ge \sum_{i=1}^n\Pr_{v \sim (F_1, \dots, F_n)}[v_i \geq x]\Pr_{v \sim (F_1, \dots, F_n)}[\max_{j\in [n]}v_j<\tau_{\ell}]\nonumber\\
    &=(1-p_{\ell})\sum_{i=1}^n\Pr_{v \sim (F_1, \dots, F_n)}[v_i \geq x]\nonumber\\
    &\ge (1-p_{\ell})\Pr_{v \sim (F_1, \dots, F_n)}[\max_{i\in [n]}v_i\ge x],\label{lem1-case2}
\end{align}
where the first inequality holds since $x\ge \tau_{\ell}$;
the second inequality holds since the summation in the third line is made of disjoint events whose union has a probability that lower bounds the probability of the second line;
the first equality holds by independence across the values;
the third inequality holds since the upper bound on $\max_{j\in [n]}v_j$ implies the corresponding event in the third line,
and the last inequality holds by the union bound.

Then, when $\ell(\boldsymbol{\tau},x)=1$, from \eqref{lem1-summation} and \eqref{lem1-case2} we get 
\begin{align*}
    \Pr_{v \sim (F_1, \dots, F_n)^k}[\algt{{\boldsymbol \tau}} (v)\geq x]&\ge \Pr_{v \sim (F_1, \dots, F_n)}[\max_{i\in [n]}v_i\ge x] \sum_{\ell=1}^{k}(1-p_{\ell})\prod_{j=0}^{\ell-1}(1-p_j)\\
    &=\Pr_{v \sim (F_1, \dots, F_n)}[\max_{i\in [n]}v_i\ge x]\sum_{\ell=1}^{k}\prod_{j=0}^{\ell}(1-p_j)\\
    &=\Phi_1(F_1,\ldots,F_n,\boldsymbol\tau)\Pr_{v \sim (F_1, \dots, F_n)}[\max_{i\in [n]}v_i\ge x].
\end{align*}
When $\ell(\boldsymbol{\tau},x)\in \{2,\ldots,k\}$, from \eqref{lem1-summation}, \eqref{lem1-case1} and \eqref{lem1-case2} we get
\begin{align*}
    &\Pr_{v \sim (F_1, \dots, F_n)^k}[\algt{{\boldsymbol \tau}} (v)\geq x]\\
    &\ge\left(\sum_{\ell=1}^{\ell(\boldsymbol{\tau},x)-1}\frac{p_\ell}{p_{\ell(\boldsymbol\tau,x)}}\prod_{j=0}^{\ell-1}(1-p_j)+\sum_{\ell=\ell(\boldsymbol{\tau},x)}^{k}(1-p_{\ell})\prod_{j=0}^{\ell-1}(1-p_j)\right)\Pr_{v \sim (F_1, \dots, F_n)}[\max_{i\in [n]}v_i\ge x]\\
    &=\left(\frac{1}{p_{\ell(\boldsymbol\tau,x)}}\sum_{\ell=1}^{\ell(\boldsymbol{\tau},x)-1}p_{\ell}\prod_{j=0}^{\ell-1}(1-p_j)+\sum_{\ell=\ell(\boldsymbol{\tau},x)}^{k}\prod_{j=0}^{\ell}(1-p_j)\right)\Pr_{v \sim (F_1, \dots, F_n)}[\max_{i\in [n]}v_i\ge x]\\
    &=\Phi_{\ell(\boldsymbol{\tau},x)}(F_1,\ldots,F_n,\boldsymbol\tau)\Pr_{v \sim (F_1, \dots, F_n)}[\max_{i\in [n]}v_i\ge x].
\end{align*}
Finally, when $\ell(\boldsymbol{\tau},x)=k+1$, from \eqref{lem1-summation} and \eqref{lem1-case1} we get
\begin{align*}
    \Pr_{v \sim (F_1, \dots, F_n)^k}[\algt{{\boldsymbol \tau}} (v)\geq x]&\ge \Pr_{v \sim (F_1, \dots, F_n)}[\max_{i\in [n] }v_i \geq x] \sum_{\ell=1}^{k}\frac{p_\ell}{p_{k+1}}\prod_{j=0}^{\ell-1}(1-p_j)\\
    &=\Pr_{v \sim (F_1, \dots, F_n)}[\max_{i\in [n] }v_i \geq x] \sum_{\ell=1}^{k}p_{\ell}\prod_{j=0}^{\ell-1}(1-p_j)\\
    &=\Phi_{k+1}(F_1,\ldots,F_n,\boldsymbol\tau) \Pr_{v \sim (F_1, \dots, F_n)}[\max_{i\in [n]}v_i\ge x],
\end{align*}
where the first equality holds since $p_{k+1}=1$.

Overall, we conclude that for every $x\ge 0$ we have
\begin{align*}
     \Pr_{v \sim (F_1, \dots, F_n)^k}[\algt{{\boldsymbol \tau}} (v)\geq x]&\ge \Phi_{\ell(\boldsymbol{\tau},x)}(F_1,\ldots,F_n,\boldsymbol\tau) \Pr_{v \sim (F_1, \dots, F_n)}[\max_{i\in [n]}v_i\ge x]\\
     &\ge \min_{i\in \{1,\ldots,k+1\}}\Phi_{i}(F_1,\ldots,F_n,\boldsymbol\tau)\Pr_{v \sim (F_1, \dots, F_n)}[\max_{i\in [n]}v_i\ge x]\\
     &=\Phi(F_1,\ldots,F_n,\boldsymbol\tau)\Pr_{v \sim (F_1, \dots, F_n)}[\max_{i\in [n]}v_i\ge x],
\end{align*}
which finishes the proof of the lemma.
\end{proof}

Since Lemma \ref{lem:dynamic-1} provides an approximate stochastic dominance inequality parameterized in the thresholds $\boldsymbol{\tau}=(\tau_1,\ldots,\tau_k)$,
in order to get the best possible approximation factor by using this approach we have to solve the following max-min problem:
\begin{align}
&\max\Big\{\textstyle\min_{i\in \{1,\ldots,k+1\}}\Phi_i(F_1,\ldots,F_n,\boldsymbol\tau):\tau_1>\tau_2>\cdots>\tau_k\ge 0\Big\} \nonumber\\
&=\max\Big\{\Phi(F_1,\ldots,F_n,\boldsymbol\tau):\tau_1>\tau_2>\cdots>\tau_k\ge 0\Big\}. \label{max-min-k}
\end{align}
In the following subsection we study this max-min problem.

\subsection{Lower Bound on the Max-Min Problem via an Explicit Feasible Solution}
\label{sec:feasible}

When $k=2$, we can compute exactly the optimal solution of the max-min problem \eqref{max-min-k}.
The problem becomes much harder for general $k$, but our following lemma establishes a double-exponentially fast increasing  
lower bound on the optimal value of \eqref{max-min-k}.
This lower bounds holds by providing a specific set of thresholds, obtained after finding a well-chosen set of quantiles $p_1,\ldots,p_k$.

\begin{lemma}\label{lem:dynamic-2}
For every $F_1,F_2,\ldots,F_n\in \Delta$ and every $k\ge 2$, let $\overline{\boldsymbol\tau}=(\overline{\tau}_1,\ldots,\overline{\tau}_k)$ such that 
$$\Pr_{v \sim (F_1, \dots, F_n)}[\max_{j\in [n]}v_j \geq \overline{\tau}_{\ell}]=1-(5/4)^{-(5/4)^{\ell}}$$ for every $\ell\in \{1,\ldots,k\}.$
Then, we have
$\Phi(F_1,\ldots,F_n,\overline{\boldsymbol\tau})\ge 1-(5/4)^{-(5/4)^k}.$
\end{lemma}

\begin{proof}
To prove the lemma, we show that the following holds: 
\begin{align*}
\Phi_i(F_1,\ldots,F_n,\overline{\boldsymbol\tau})&\ge 1\;\text{ for every $i\in \{1,\ldots,k\}$, and }\\
\Phi_{k+1}(F_1,\ldots,F_n,\overline{\boldsymbol\tau})&\ge 1-(5/4)^{-(5/4)^k}.
\end{align*}
Then, the lemma follows since, overall, we get \begin{equation*}\Phi(F_1,\ldots,F_n,\overline{\boldsymbol\tau})\ge
1-(5/4)^{-(5/4)^k}.\end{equation*}
Let  $p_0=0$, and $p_{\ell} =\Pr_{v \sim (F_1, \dots, F_n)}[\max_{j\in [n]}v_j \geq \overline{\tau}_{\ell}]$ for every $\ell\in \{1,\ldots,k+1\}$.
When $i=1$, we have 
\begin{align*}
    \Phi_1(F_1,\ldots,F_n,\overline{\boldsymbol\tau})&=\sum_{\ell=1}^{k}\prod_{j=1}^{\ell}(1-p_j)\\
    &=\sum_{\ell=1}^{k}\prod_{j=1}^{\ell}(5/4)^{-(5/4)^{j}}\\
    &=\sum_{\ell=1}^k(5/4)^{5(1-(5/4)^{\ell})}\\
    &\ge (5/4)^{5(1-(5/4))}+(5/4)^{5(1-(5/4)^{2})}> 1,
\end{align*}
where the first inequality holds since $k\ge 2$.
For $i=2$, 
\begin{align*}
\Phi_2(F_1,\ldots,F_n,\overline{\boldsymbol\tau}) 
& = \frac{p_{1}}{p_2} + \sum_{\ell=2}^{k}\prod_{j=0}^{\ell}(1-p_j) \geq \frac{p_{1}}{p_2} + (1-p_1)(1-p_2) \\ &=  \frac{1-(5/4)^{-5/4}}{1-(5/4)^{-(5/4)^2}} + (5/4)^{-5/4} \cdot (5/4)^{-(5/4)^2} > 1,
\end{align*}
where the first inequality is since $k\geq 2$.

For every $i\in \{3,\ldots,k+1\}$, we have 
\begin{align*}
\sum_{\ell=1}^{i-1}p_{\ell}\prod_{j=0}^{\ell-1}(1-p_j)
&=\sum_{\ell=1}^{i-1}\left(\prod_{j=0}^{\ell-1}(1-p_j)-(1-p_{\ell})\prod_{j=0}^{\ell-1}(1-p_j)\right)\\
&=\sum_{\ell=1}^{i-1}\left(\prod_{j=0}^{\ell-1}(1-p_j)-\prod_{j=0}^{\ell}(1-p_j)\right)\\
&=1-\prod_{j=0}^{i-1}(1-p_j)\\
&=1-\prod_{j=1}^{i-1}(5/4)^{-(5/4)^{j}}= 1-(5/4)^{-\sum_{j=1}^{i-1}(5/4)^j} = 1-(5/4)^{5(1-(5/4)^{i-1})}.
\end{align*}
Thus, for every $i\in \{3,\ldots,k\}$ we have 
\begin{equation*}
\Phi_i(F_1,\ldots,F_n,\overline{\boldsymbol\tau})\ge \frac{1}{p_i}\sum_{\ell=1}^{i-1}p_{\ell}\prod_{j=0}^{\ell-1}(1-p_j)=\frac{1-(5/4)^{5(1-(5/4)^{i-1})}}{1-(5/4)^{-(5/4)^{i}}}\ge 1,
\end{equation*}
where the last inequality holds since $i\geq 3$, 
and when $i=k+1$, we have 
\begin{equation*}\Phi_{k+1}(F_1,\ldots,F_n,\overline{\boldsymbol\tau})=\sum_{\ell=1}^{k}p_{\ell}\prod_{j=0}^{\ell-1}(1-p_j)=1-(5/4)^{5(1-(5/4)^{k})}\ge 1-(5/4)^{-(5/4)^k},\label{case-k+1}
\end{equation*} 
where the last inequality 
holds for every $k\geq 2$.
This concludes the proof.
\end{proof}

\subsection{Putting it All Together}
\label{sec:full-proof}

With Lemma~\ref{lem:dynamic-1} and Lemma~\ref{lem:dynamic-2} at hand, we are now ready to prove our main theorem.

\begin{proof}[Proof of Theorem \ref{thm:dynamic}]
Given $\varepsilon>0$, by Lemma \ref{lem:dynamic-1} and Lemma \ref{lem:dynamic-2}, we have that for every $F_1,\ldots,F_n\in \Delta$ and every $k\ge \max(2,\log_{5/4}\log_{5/4}(1/\varepsilon))$, by taking $\overline{\boldsymbol\tau}=(\overline{\tau}_1,\ldots,\overline{\tau}_k)$ as defined in Lemma \ref{lem:dynamic-2}, the following holds:
\begin{align*}
    \EE_{v \sim (F_1, \dots, F_n)^k}[\algt{{\overline{\boldsymbol \tau}}} (v)]&=\int_{0}^{\infty}\Pr_{v \sim (F_1, \dots, F_n)^k}[\algt{{\overline{\boldsymbol \tau}}} (v)\geq x]\mathrm dx \\
    &\ge \Phi(F_1,\ldots, F_n,\overline{\boldsymbol\tau}) \int_{0}^{\infty}\Pr_{v \sim (F_1, \dots, F_n)}[\max_{i \in [n]} v_i \geq x]\mathrm dx\\
    &\ge \Big(1-(5/4)^{-(5/4)^k}\Big)\int_{0}^{\infty}\Pr_{v \sim (F_1, \dots, F_n)}[\max_{i \in [n]} v_i \geq x]\mathrm dx\\
    &\ge (1-\varepsilon)\int_{0}^{\infty}\Pr_{v \sim (F_1, \dots, F_n)}[\max_{i \in [n]} v_i \geq x]\mathrm dx\\
    &=(1-\varepsilon)\cdot \EE_{v \sim (F_1, \dots, F_n)}[\max_{i \in [n]} v_i],
\end{align*}
which implies that the $(1-\varepsilon)$-competition complexity of the class of block threshold algorithms is $O(\log\log(1/\varepsilon))$.

Then, the theorem follows by the fact that \cite[Theorem 2.2]{BrustleCDV22} show that the $(1-\varepsilon)$-competition complexity of the class of general
threshold algorithms when $F_1=F_2=\cdots=F_n$ is $\Omega(\log\log(1/\varepsilon))$.  
\end{proof}

We note that our proof of Theorem~\ref{thm:dynamic} actually shows a $\Theta(\log \log (1/\varepsilon))$ competition complexity guarantee in an approximate stochastic dominance sense, which is stronger than the regular competition complexity result, based on comparing expectations. 
In particular, our argument strengthens the result of Brustle et al. \cite{BrustleCDV22} even in the i.i.d.~case.

We know from the lower bound on the competition complexity of \cite{BrustleCDV22} that one cannot strengthen Lemma~\ref{lem:dynamic-2} to show that for every $c$ one can devise a series of quantiles for which $\Phi \geq 1- c^{-c^k} $. However, the proof in \cite{BrustleCDV22} is not explicit, and we now present a simpler proof that shows that one cannot obtain better than $1-\Omega(c^{-c^k})$ approximation to the value of the prophet using $k$ blocks in the stochastic dominance sense for $c=3$.

\begin{proposition}
For every $k\geq 1$, there exists a positive integer value $n_k$ and $F_1,\ldots,F_{n_k}\in \Delta$, such that for every $\boldsymbol\tau=(\tau_1,\ldots_,\tau_k)$, there exists $x\in \reals_{\geq 0}$ such that 
$$\Pr_{v\sim (F_1,\ldots,F_{n_k})^k}[\algt{\boldsymbol{\tau}}(v) \geq x] <(1-\varepsilon_k)\Pr_{v\sim (F_1,\ldots,F_{n_k})}[\max_{i \in [n_k]}v_i  \geq x] ,$$
where $\varepsilon_k=3^{-3^k}$ for every $k$.
\end{proposition}
\begin{proof}
For every $k\ge 1$, let $n_k=3^{3^{k+2}}$, and for every $i\in [n_k]$, let $F_i$ be the distribution of the random variable $v_i=i \cdot \text{Bernoulli}(1/i) - \text{Uniform}[0,1]$. 
For ease of notation, let $F=(F_1,\ldots,F_{n_k})$.
We first observe that since $\max_{i\in[n_k]} i \cdot \text{Bernoulli}(1/i)$ is distributed uniformly on the set $\{1,\ldots,n_k\}$, the random variable $\max_{i\in [n_k]} v_i$ is  uniformly distributed over the interval $[0,n_k]$, and for every $x\in [0,n]$, it holds that $ \Pr[\max_{i\in [n_k]} v_i \geq x] = (n-x)/n$. 
For this instance, the optimal block threshold algorithm uses monotone decreasing thresholds, since otherwise sorting them leads to thresholds that stochastically dominates the original ones.
Given $\boldsymbol{\tau}=(\tau_1,\ldots,\tau_k)$, let $A = \{i\in [k] : \tau_i   \leq    n_k\cdot  3^{-3^i}\}$.
If $A = \emptyset$, then  for $x=0$, it holds that 
\begin{align*}
\Pr_{v\sim F^k}[\algt{\boldsymbol{\tau}}(v)] \geq x] & =   1- \prod_{\ell=1}^k \Pr_{v^{(\ell)}\sim F}[\max_{i \in [n_k]}v_i^{(\ell)} <  \tau_\ell] \\  &\leq 1- \prod_{\ell=1}^k \frac{\tau_\ell }{n_k} \leq 1- \prod_{\ell=1}^k 3^{-3^\ell }<
(1-\varepsilon_k)\Pr_{v\sim F}[\max_{i \in [n_k]}v_i  \geq x] 
\end{align*}
where 
the first inequality holds since $\max_{i\in [n_k]} v_i^{(\ell)}$ is uniformly distributed on the interval $[0,n]$, the second inequality is since $A=\emptyset$, and the last inequality holds by the definition of $\varepsilon_k$, together with $\Pr_{v\sim F}[\max_{i \in [n_k]}v_i  \geq x] =1$.

Else, if $A \neq \emptyset$, let $i^* = \min A$, and let $x=n_k\cdot  3^{-2\cdot3^{i^*-1}}$.
Then, we denote by $B_\ell$ (and by $\overline{B}_{\ell}$ its complement) the event in which there exists a value in block $\ell$ exceeding the threshold $\tau_\ell$, and by $B^x_\ell$ the event that in block $\ell$, for all $i\in [n_k]$ it holds that $v_i^{(\ell)} \notin [\tau_\ell,x) $. For $\ell\geq i^*$, it holds that 
\begin{equation}
\label{eq:blx}    
\Pr[B_\ell^x \text{ and } B_\ell] = \frac{\tau_\ell- \lfloor \tau_\ell \rfloor}{\tau_\ell}\cdot\left(\prod_{j=\lceil \tau_\ell \rceil}^{x-1}\frac{i-1}{i}\right)\cdot \left(1- \prod_{j=x}^{n_k} \frac{i-1}{i}\right) = \frac{n_k-x}{n_k} \cdot  \frac{\tau_\ell}{x}.\end{equation}
Thus, we have
\begin{align*}
\Pr_{v\sim F^k}[\algt{\boldsymbol{\tau}}(v)] \geq x] & =  \sum_{\ell=1}^k \left(\prod_{j=1}^{\ell-1}\Pr[\overline{B}_j]\right)\cdot\Pr[B_\ell] \cdot \Pr[\algt{\boldsymbol{\tau}}(v)\geq x \mid B_{\ell},\overline{B}_1,\ldots,\overline{B}_{\ell-1}]  
\\ & \leq  \sum_{\ell=1}^{i^*-1} \left(\prod_{j=1}^{\ell-1}\Pr[\overline{B}_j]\right)\cdot\Pr[B_\ell] + \sum_{\ell=i^*}^{k} \left(\prod_{j=1}^{\ell-1}\Pr[\overline{B}_j]\right)\cdot \Pr[B_\ell^x \wedge B_\ell]
\\ & \leq  \sum_{\ell=1}^{i^*-1} \left(\prod_{j=1}^{\ell-1}3^{-3^j}\right)\cdot \Big(1-3^{-3^\ell}\Big) + \sum_{\ell=i^*}^{k} \left((3^{-3^{i^*}})^{\ell-i^*}\prod_{j=1}^{i^*-1}3^{-3^j}\right) \cdot \frac{n_k-x}{n_k} \cdot  \frac{\tau_\ell}{x} 
\\ & < (1-\varepsilon_k) \cdot \frac{n_k-x}{n_k} = (1-\varepsilon_k)\Pr_{v\sim F}[\max_{i \in [n_k]}v_i  \geq x] ,
\end{align*}
where the first inequality holds since $\Pr[\algt{\boldsymbol{\tau}}(v)\geq x \mid B_{\ell},\overline{B}_1,\ldots,\overline{B}_{\ell-1}] $ is  $1$ for $\ell< i^*$, and is $\Pr[B_\ell^x \mid B_\ell]$ otherwise; the second inequality is by the definition of $A$, since the expression is monotone decreasing in $\tau_\ell$ for $\ell<i^*$, monotone increasing in $\tau_\ell$ for $\ell \geq i^*$, and by Equation~\eqref{eq:blx}; the last inequality holds by the definitions of $\varepsilon_k$ for every $k$, and the last equality is since $\max_{i \in [n_k]}v_i $ is uniformly distributed over the interval $[0,n_k]$.
\end{proof}

%% file: static.tex
\section{Single Threshold Algorithms}
\label{sec:static}

In this section, we study the competition complexity of the class of single
threshold algorithms. 
Let $\algt{\tau}$ denote the single 
threshold algorithm with threshold $\tau \in \reals_{\geq 0}$. We show the following tight bound on the competition complexity of single threshold algorithms.
  
\begin{theorem}\label{thm:static}
For every $\varepsilon>0$ the
$(1-\varepsilon)$-competition complexity of the class of 
static 
single threshold algorithms is $\Theta\left(\log(1/\varepsilon)\right)$.
\end{theorem}

We prove Theorem~\ref{thm:static} through Lemma~\ref{lem:static} and Proposition~\ref{prop:static} below. Lemma~\ref{lem:static} establishes the upper bound claimed in the theorem, and Proposition~\ref{prop:static} shows a matching lower bound.

We begin with Lemma~\ref{lem:static}, which shows that for every instance, there exists a single threshold $\tau^*$ such that for every $\varepsilon>0$, the $(1-\varepsilon)$-competition complexity of the single threshold algorithm $\ALG_{\tau^*}$ with respect to the instance, is  $O(\log(1/\varepsilon))$. 
The upper bound follows rather directly from the celebrated ``median rule'' proof of Samuel-Cahn~\cite{SamuelCahn84}; we work a bit harder to show that it also holds in a stochastic dominance sense.

\begin{lemma}\label{lem:static}
For every $F_1,\ldots,F_n$, there exists a threshold $\tau^*$ such that for every  $\varepsilon>0$, 
and for every $k\geq \log_2(1/\varepsilon)$
\[
\EE_{v \sim (F_1, \dots, F_n)^k}[\algt{\tau^*}(v)] \geq (1-\varepsilon) \cdot \EE_{v \sim (F_1, \dots, F_n)}[\max_{i \in [n]} v_i].
\]
Furthermore, the stronger approximate stochastic-dominance inequality
\[
\Pr_{v \sim (F_1, \dots, F_n)^k}[\algt{\tau^*}(v) \geq x] \geq (1-\varepsilon) \Pr_{v \sim (F_1, \dots, F_n)}[\max_{i \in [n]} v_i \geq x]
\]
holds for all $x \in \reals_{\geq 0}$.
\end{lemma}
\begin{proof}
Consider the unique  threshold $\tau^*$ satisfying the following equation
\begin{equation}
    \prod_{i=1}^n \Pr_{v_i
\sim F_i}[v_i \leq \tau^*] = \frac{1}{2} .\label{eq:threshold}
\end{equation}

In what follows, for ease of notation we consider the input sequence of $nk$ values as $v_1,\ldots,v_{nk}$, where $v_{n(\ell-1)+i}\sim F_i$ for every $\ell\in [k]$ and every $i\in [n]$.
We next show that for every number of copies $k\geq 1$, given that the algorithm accepts a value, its expectation is larger than the value of the prophet for a single block, i.e., for every $k\ge 1$,
\begin{equation} \label{eq:ex_th}
    \EE_{v \sim (F_1, \dots, F_n)^k}[\algt{\tau^*}(v) \mid  \text{there exists }  i\in [nk]: ~v_i \geq \tau^*] \geq \EE_{v \sim (F_1, \dots, F_n)}[\max_{i \in [n]} v_i].
\end{equation}
In fact we show the following stronger claim, that for every $x\ge 0$ it holds that \begin{equation} \label{eq:ex_th-stochastic}
    \Pr_{v \sim (F_1, \dots, F_n)^k}[\algt{\tau^*}(v) \geq x \mid  \text{there exists }  i\in [nk]: ~v_i \geq \tau^*] \geq \Pr_{v \sim (F_1, \dots, F_n)}[\max_{i \in [n]} v_i \geq x].
\end{equation}

Let $\hat{i}$ be the first index of a value $v_i$ that exceeds $\tau^*$ (and $\hat{i}=0$ if no such $v_i$ exists), and let $\hat{r} = \lceil \hat{i}/n \rceil$, i.e., the block from which a value is chosen. 
The LHS of Equation~\eqref{eq:ex_th} (respectively, Equation~\eqref{eq:ex_th-stochastic})
can be rewritten as $\EE_{v \sim (F_1, \dots, F_n)^k}[\algt{\tau^*}(v) \mid  \hat{r}\neq 0]$ (respectively, $\Pr_{v \sim (F_1, \dots, F_n)^k}[\algt{\tau^*}(v)  \geq x \mid  \hat{r}\neq 0]$).
Thus, it is sufficient to prove that  Equation~\eqref{eq:ex_th} 
 holds for every non-zero realization of $\hat{r}$, i.e., for every $k\ge 1$ and every $r\in [k]$,
\begin{equation} \label{eq:ex_th2}
    \EE_{v \sim (F_1, \dots, F_n)^k}[\algt{\tau^*}(v) \mid  \hat{r}=r] \geq \EE_{v \sim (F_1, \dots, F_n)}[\max_{i \in [n]} v_i].
\end{equation}

Since the values of $v$ that don't belong to the $r$-th $n$-tuple of values can be ignored, one can observe that  Equation~\eqref{eq:ex_th2} is equivalent 
to the proof of the original prophet inequality, provided by Samuel-Cahn \cite{SamuelCahn84}. We next prove the stronger claim of Equation~\eqref{eq:ex_th-stochastic}.
To this end, we show that for every $x \in
\mathbb{R}_{\geq 0}$, and  for $r\neq 0$  it holds that 
$$   \Pr_{v \sim (F_1, \dots, F_n)^k}[\algt{\tau^*}(v) \geq x \mid  \hat{r}=r] \geq \Pr_{v \sim (F_1, \dots, F_n)}[\max_{i \in [n]} v_i \geq x]. $$
Note, that for $x\leq \tau^*$, the LHS is $1$, thus the inequality holds, and it is sufficient to prove it for $x>\tau^*$. Thus, 

\begin{align*}
  &\Pr_{v \sim (F_1, \dots, F_n)^k}[\algt{\tau^*}(v) \geq x \mid  \hat{r}=r] \\ 
  & =  
  \sum_{i=(r-1)n+1}^{rn} \Pr[\text{for all }{j \in [i-1]}: v_j < \tau^*,  \text{ and }  v_i>x  \mid \hat{r}=r]  
   \\
  &=\sum_{i=(r-1)n+1}^{rn} \Pr[\text{for all }{j \in [i-1]}: v_j < \tau^*, v_i>x  \text{ and } \hat{r}=r] \cdot \frac{1}{\Pr[\hat{r}=r]} \\
  &=2^r \sum_{i=(r-1)n+1}^{rn} \Pr[\text{for all }{j \in [i-1]}: v_j < \tau^*,  \text{ and } v_i>x ]  \\
  &\ge 2^r\sum_{i=(r-1)n+1}^{rn} \frac{1}{2^{r}}\cdot \Pr[ v_i>x ]  \geq  \Pr_{v \sim (F_1, \dots, F_n)}[\max_{i \in [n]} v_i \geq x],
\end{align*}
where the first equality is since $x>\tau^*$, the third equality is since by definition of $\tau^*$, $\hat{r}=r$ with probability $1/2^r$ and since $\hat{r}=r$ follows from the other events in the probability, the first inequality is since the  $\Pr[\text{for all } j\in\{1,\ldots (r-1)n\}: v_j <\tau^*] = 1/2^{r-1}$, and since  $\Pr[\text{for all } j\in \{(r-1)n+1,\ldots, i-1\}: 
v_j <\tau^*] \geq 1/2$, and all the events are independent. The last inequality follows since the sum of probabilities that $v_i$ exceeds $x$, is at least the probability that the maximum exceeds $x$.

The proof of the lemma then follows by combining Equation~\eqref{eq:ex_th} (respectively for the ``furthermore'' part, Equation~\eqref{eq:ex_th-stochastic}) with the observation that, 
by definition of $\tau^*$, it holds that 
\begin{align*}
\Pr[ \text{there exists } i\in [nk]: ~v_i \geq \tau^*] &  = 1-\Pr[ \text{for all } i\in [nk]: ~v_i \leq \tau^*] \\ 
&=  1-\Pr[ \text{for all } i\in [n]: ~v_i \leq \tau^*]^k =1-\frac{1}{2^k} \geq 1-\varepsilon,
\end{align*}
where the inequality holds  for $k\geq \log_2(1/\varepsilon)$.
\end{proof}

Our next result, Proposition~\ref{prop:static}, shows a matching lower bound of $\Omega(\log(1/\varepsilon))$ on the competition complexity of single
threshold algorithms. This lower bound holds even with respect to the case of i.i.d.~distributions.

\begin{proposition} \label{prop:static}
For every $n\geq 2$, and for every $\varepsilon\in(0,1)$ there exists an instance with $n$ i.i.d. values that are distributed according to some distribution $F$, such that for every $\tau$, and every $k  < \frac{\log_2  (1/\varepsilon)}{804}$, it holds that
\begin{equation} \label{eq:upper-static}
    \EE_{v \sim F^{k\cdot n}}[\algt{\tau}(v)] < (1-\varepsilon) \cdot \EE_{v \sim F^n}[\max_{i \in [n]} v_i].
\end{equation}
\end{proposition}
\begin{proof}
It is sufficient to consider $\varepsilon< 1/20$, since otherwise $k=0$, and the claim holds trivially.
Consider the distribution $F$ in which $v_i = 1+200\varepsilon \cdot \text{Bernoulli}\left(1/(20 n)\right) +\text{Uniform}[0,\varepsilon]$  for every $i\in [n]$ and every $\ell \in [k]$.
The RHS of Equation~\eqref{eq:upper-static} satisfies that
\begin{equation}
     (1-\varepsilon) \cdot \EE_{v \sim F^n}[\max_{i \in [n]} v_i] \geq (1-\varepsilon) \cdot \Big(1 + 200\varepsilon (1-e^{-1/20}) \Big) \geq 1+ 8\varepsilon, \nonumber
\end{equation}
where the first inequality is since the maximum is at least $1+200\varepsilon$ with probability greater than $1-e^{-1/20}$, and otherwise it is at least $1$. The second inequality holds for every $\varepsilon \in(0,1/20)$.
Now consider a static threshold algorithm with a threshold $\tau$.\\

\noindent \textbf{Case 1:}
If $\tau> 1+\frac{(20n -29) \varepsilon}{20 n}$, then the LHS of Equation~\eqref{eq:upper-static} satisfies that 
\begin{align*}
\EE_{v \sim F^{k\cdot n}}[\algt{\tau}(v)]  & \leq  \Pr[\text{there exists } i\in [nk]: v_i \geq \tau] \cdot (1+201\varepsilon) \\ 
& \leq  (1-(1-3/2n)^{nk})(1+201\varepsilon) \leq (1-16^{-k}) (1+201\varepsilon) \leq 1,
\end{align*}
where the first inequality is since the support of $F$ is bounded by $1+201\varepsilon$, the second inequality is since the probability that $v_i>\tau$ is at most $3/(2n)$, the third inequality is since $n\geq 2$, and the last inequality is since $k  < \frac{\log_2  (1/\varepsilon)}{804} $.
Thus, Equation~\eqref{eq:upper-static} holds.\\

\noindent \textbf{Case 2:} If $\tau \leq  1+\frac{(20 n -29) \varepsilon}{20 n}$, then let $q= \frac{1+\varepsilon-\tau}{\varepsilon} \geq \frac{29}{20n}$. The LHS of Equation~\eqref{eq:upper-static} satisfies that 
\begin{align*}
\EE_{v \sim F^{k\cdot n}}[\algt{\tau}(v)]  & \leq  \EE_{v_i \sim F}[v_i \mid v_i \geq \tau ]\\
&= \EE_{v_i \sim F}[v_i \cdot  \mathds{1}[ v_i \geq \tau] ] /\Pr[v_i \geq \tau]  \\ 
& \leq  \left(\frac{1}{20 n}(1+201\varepsilon) + \frac{20n-1}{20n}q(1+\varepsilon)\right) /\left( \frac{1}{20n} + \frac{20n-1}{20n}q\right) \\ 
& =  \left((1+201\varepsilon) + (20n-1)q(1+\varepsilon)\right) /\left( 1 + (20n-1)q\right)\\ 
&  <  1+8\varepsilon,
\end{align*}
where the last inequality holds for every $n\geq 2$, and $q \geq \frac{29}{20n}$.
Thus, Equation~\eqref{eq:upper-static} holds.
\end{proof}

%% file: extensions.tex
\section{Combinatorial Extensions}\label{sec:extensions}

In this section, we present a generalization of our model to combinatorial settings, and discuss some initial results. We present the general model in Section~\ref{sec:general-model}, and some preliminary results for combinatorial auctions in Section~\ref{sec:xos}. In Section~\ref{app:matching} we give additional results for bipartite matching with one-sided vertex arrivals. 
This latter result is actually implied by the result in Section~\ref{sec:xos}. The purpose of presenting an alternative proof
is to show how a different technique, in this case online contention resolution schemes, can also be used to study the competition complexity in combinatorial settings.

\subsection{A General Model}\label{sec:general-model}

In online combinatorial Bayesian selection problems, there is a series of  $n$ decisions that need to be made. Each decision $i\in[n]$ is associated with a set of alternatives $A_i$ from which the decision-maker needs to choose, and with an information $v_i$ drawn independently from some $F_i$ on support $S_i$ that is revealed to the decision-maker at the time of decision $i$. We denote by $S=\bigtimes_{i \in [n]} S_i$, $F= \bigtimes_{i \in [n]} F_i$, and $A= \bigtimes_{i \in [n]} A_i$.
Additionally,  there is a non-empty feasibility constraint $\feas \subseteq A$, such that the decision-maker, must select a tuple of alternatives $(a_1,\ldots,a_n) $ that is in $\feas$, and there is a reward function $f: A \times S \rightarrow \reals_{\geq 0}$. 

At each step $i \in [n]$, the algorithm
observes $v_i$, and needs to select an alternative $a_i \in A_i$ in 
an immediate and  irrevocable way.  
The algorithm's performance is measured against the offline optimum.
By fixing a class of feasibility constraints $\mathbb{C}_\feas$, a class of valuation functions $\mathbb{C}_f$, and a class of distributions $\mathbb{C}_F$,
an online algorithm $\ALG$ is $\alpha$-competitive if
\[
\inf_{F\in \mathbb{C}_F}  \inf_{\feas \in\mathbb{C}_\feas} \inf_{f\in \mathbb{C}_f} \frac{\EE_{v \sim F}[f(\ALG(v),v) \cdot \indicator{\ALG(v) \in \feas}]}{\EE_{v \sim F}[\max_{a\in \feas} f(a,v)]} \geq \alpha,
\]
where $\ALG(v)$ is the (possibly random) tuple of alternatives chosen by the algorithm $\ALG$ when observing a sequence of information $v$.

A simple example of this setting is when for all $i\in [n]$ we have
$A_i= \{0,1\}$, 
$\mathbb{C}_\feas = \{\{(a_1,\ldots,a_n) \mid \sum_{i\in [n]} a_i \leq 1\}\}$,  $\mathbb{C}_f = \{\{ f(a,v) = \sum_{i \in [n]} a_i \cdot v_i\}\}$,  and $\mathbb{C}_F = \Delta^n$, which corresponds to the standard (single-choice) prophet inequality setting.

We next describe a
generalization of the block model to the combinatorial Bayesian selection framework. 
The input to the algorithm is given by $k$ copies of an online combinatorial Bayesian selection problem, so there are $kn$ decisions in total. For every $j \in [kn]$, the $j$-th decision is of type $i$ if 
$j\equiv i\mod n$,
and for each decision $j$ of type $i$, the information $v_j$ is sampled independently according to $F_i$. 
For every decision $j$ of type $i$ we must select an alternative in $A_i$.

The output of the algorithm is a $kn$-dimensional vector of alternatives.
To define feasibility, and to evaluate the reward achieved by an output, we require that there is an infinite series of classes of feasibility constraints $(\mathbb{C}^i_\feas)_{i\in \mathbb{N}}$, and an infinite series of classes of reward functions $(\mathbb{C}^i_f)_{i \in \mathbb{N}}$ so that we can evaluate feasibility via $\feas_k \in \mathbb{C}^k_\feas$ and the reward via $f_k\in \mathbb{C}^k_f$. For a concrete problem it is typically clear how to define $(\mathbb{C}^i_\feas)_{i\in \mathbb{N}}$ and $(\mathbb{C}^i_f)_{i \in \mathbb{N}}$ based on $\mathbb{C}_\feas$ and $\mathbb{C}_f$.

We are interested in comparing the expected reward of the algorithm on $k$ copies to the expected optimal reward on a single copy.

\begin{definition}[Combinatorial competition complexity] \label{def:comb}
Given a series of decisions associated with alternatives $A=A_1\times\ldots\times A_n$, a class of distributions $\mathbb{C}_F$, an infinite series of classes of feasibility constraints $(\mathbb{C}^i_\feas)_{i\in \mathbb{N}}$, and an infinite series of classes of reward functions $(\mathbb{C}^i_f)_{i \in \mathbb{N}}$,   
for every $\varepsilon \ge 0$, the $(1-\varepsilon)$-competition complexity with respect to a class of algorithms $\mathcal{A}$ is the smallest positive integer number $k(\varepsilon)$ such that for every $k\geq k(\varepsilon)$, every $F\in \mathbb{C}_F$, $\feas_k\in \mathbb{C}^k_\feas$, $f_k\in \mathbb{C}^k_f$, it holds that
\begin{equation}
\label{eq:defcomb}    
\max_{\ALG \in \mathcal{A}_{n,k}}\EE_{v \sim F^k}[f_k(\ALG(v),v)\cdot \indicator{\ALG(v) \in \feas_k}] \geq (1-\varepsilon) \cdot \EE_{v \sim F}[\max_{a \in \feas_1} f_1(a,v)],
\end{equation}
where $\mathcal{A}_{n,k}$ are all algorithms in $\mathcal{A}$ that are defined on $\mathbb{C}_F$ and $\mathbb{C}^k_\feas$, and $\ALG(v)$ is the (possibly random) tuple of alternatives chosen by $\ALG$ when observing a sequence of values $v$. 
\end{definition}

For example, to obtain the competition complexity for the  standard prophet inequality setting, 
we can let 
$A_i= \{0,1\}$ for all $i\in [n]$, 
$\mathbb{C}^k_\feas = \{\{(a_1,\ldots,a_{n\cdot k}) \mid \sum_{i\in [n\cdot k]} a_i \leq 1\}\}$,  $\mathbb{C}^k_f = \{\{ f(a,v) = \sum_{i \in [n\cdot k]} a_i \cdot v_i\}\}$,  and $\mathbb{C}_F = \Delta^n$.

\subsection{Combinatorial Auctions}\label{sec:xos}

In this section, we generalize the static pricing scheme by Feldman et al.~\cite{FeldmanGL15} for XOS markets, to a dynamic pricing scheme that has a $(1-\varepsilon)$-competition complexity  of $O\left(\log(1/\varepsilon)\right)$, and a static pricing scheme that has a $(1-\varepsilon)$-competition complexity  of $O\left(1/\varepsilon\right)$.

In the combinatorial auction setting, there is a set $M$ of $m$ items, and $n$ agents. Each agent $i\in [n]$ is associated with a valuation function $v_i:2^M\rightarrow \reals_{\geq 0} $ that is drawn independently from a distribution $F_i$. We consider price-based algorithms, which set a price $p_j$ for each item $j \in M$. 
The agents then arrive one-by-one, and purchase a set of available items in their demand. That is, agent $i$ buys a set $T_i$ that maximizes the utility $v_i(T) - \sum_{j \in T} p_j$ over all $T \subseteq M'$, where $M' \subseteq M$ are the items that are available when agent $i$ arrives. 
We evaluate the performance of a pricing scheme by the expected social welfare it achieves, that is  $\EE_{v \sim F}[\sum_{i \in [n]} v_i(T_i)]$.

We consider a repeated version of the combinatorial auction problem where we see $k$ independent copies of the buyers. As before we refer to each copy as a block. The valuation of buyer $i \in [nk]$ of type $r \in [n]$ is drawn independently according to the distribution $F_{r}$. Each buyer $i \in [nk]$ purchases a set of available items $T_i$ that maximizes their utility. 
We compare the expected social welfare of a price-based algorithm on $k$ copies to the expected maximum social welfare on a single copy.

A pricing scheme is called \emph{static} if the prices of the items are fixed in advance before the arrival of the agents, and is called \emph{dynamic}, if the prices may adapt after each agent has made a purchase (before the arrival of the next agent).
We define {\it block-consistent prices} as prices that are static throughout each block but can change between blocks.

A valuation function $v:2^M \rightarrow \reals_{\geq 0}$ is called XOS (a.k.a. fractionally subadditive) if there exists a non-empty set of additive functions $G= \{g_1,\ldots g_s\}$ for some positive integer $s$, 
such that $v(T)=\max_{t\in [s]} \sum_{j\in T} g_t(j)$ for every $T\subseteq M$. 
A valuation function $v: 2^M \rightarrow \reals_{\geq 0}$ is submodular if for every $T, T'$ with $T \subseteq T'$ and every $j \in M \setminus T' \subseteq M$ it holds that $v_i(T \cup \{j\}) - v_i(T) \geq v_i(T' \cup \{j\}) - v_i(T')$. Every submodular valuation function is XOS, but not vice versa \cite{LehmannLN06}.
We refer to combinatorial auctions in which all valuation functions are submodular (resp.~XOS), as submodular (resp.~XOS) combinatorial auctions.

\begin{theorem}\label{thm:xos-dynamic}
For every $k\geq 1$, and every XOS combinatorial auction defined by $M$, $n$, and a product distribution $F=F_1\times\ldots\times F_n$,  there exists a block-consistent pricing scheme on $k$ copies such that 
$$ \EE_{v\sim F^k}\Big[\sum_{i\in [n k]} v_i(T_i)\Big] \geq  \left(1-\frac{1}{2^k}\right) \cdot \EE_{v\sim F}\Big[\max_{(T_1, \ldots, T_n) \in P(M)} \sum_{i\in [n]} v_i(T_i)\Big],$$
where $T_i$ is the demanded set purchased by each agent $i\in [n]$, and $P(M)$ is the set of all partitions of $M$ into $n$ sets. In particular, the $(1-\varepsilon)$-competition complexity of block-consistent prices for submodular and XOS combinatorial auctions is $O(\log(1/\varepsilon))$.
\end{theorem}
\begin{proof}
For every item $j\in M$, and agent $i$ let $\OPT_{j,i} = E_{v\sim F }[g_{i(j)}(j) \cdot \indicator{i=i(j)}]$, where $i(j)$ is the agent that receives item $j$ in the welfare-maximizing  allocation according to valuation profile $v=(v_1,\ldots,v_n)$, and $g_{i(j)}$ is the additive function that corresponds to the definition of XOS valuation that maximizes the value of the set agent $i(j)$ receives in the optimal allocation according to valuation profile $v=(v_1,\ldots,v_n)$.

Let $p$ be the the block-consistent pricing scheme where the price during block $t\in [k]$  for item $j\in M$ is $p_{t,j} = (1-\frac{1}{2^{k+1-t}})\cdot \sum_{i\in [n]} \OPT_{j,i}$.
Let $X_j$ be the block in which item $j$ is purchased (if it is not purchased, then let $X_j=k+1$), and 
let $q_{t,j}=\Pr[X_j = t]$.

The welfare of the mechanism can be decomposed into revenue and surplus. 
For the revenue, we have
$$\EE_{v\sim F^k}[\text{Revenue}] =\sum_{j\in M}\sum_{t\in [k]}q_{t,j} \cdot p_{t,j} = \sum_{j\in M}\sum_{t\in [k]}q_{t,j} \cdot \left(1-\frac{1}{2^{k+1-t}}\right)\cdot \sum_{i \in [n]} \OPT_{j,i} . $$
For every agent $i\in [nk]$, let's denote by $\AV_i$ the set of items that is available when $i$ arrives. For each type $r \in [n]$ let's denote by $T_{r}^*$ the set that agent $i$ of type $r$ receives in the welfare-maximizing allocation on profile $v \sim F$, and by $g_{r}^*$ the additive function that maximizes agent $r$'s value with respect to $v_i$ and $T_{r}^*$. Then, if agent $i$ is of type $r$,

\begin{align*} 
\EE_{v\sim F^k}[\Surplus_i] & =    \sum_{M'\subseteq M} \Pr[\text{AV}_i = M'] \cdot \EE_{v_i\sim F_{r}} \Big[\max_{T\subseteq M' }(v_i(T) -p(T))\Big] \\ 
& \geq  \sum_{M' \subseteq M} \Pr[\AV_i = M'] \cdot \EE_{v\sim F} [ v_{r}(T^*_{r}\cap M') -p(T^*_{r} \cap M')] \\ 
& \geq   \sum_{M'\subseteq M} \Pr[\AV_i = M'] \cdot \EE_{v\sim F} \Big[ \sum_{j \in T^*_{r}\cap M'} (g_{r}^*(j) -p(j))\Big] \\
& \geq \sum_{j\in M} \Big(1-\sum_{t=1}^{\lceil i/n\rceil}q_{j,t} \Big)\cdot 
\frac{\OPT_{j,r}}{2^{k+1-\lceil\frac{i}{n}\rceil} },
\end{align*}
where the first inequality is since $v_i$ and $v_{r}$ are drawn from the same distribution, and since $T_i$ is the demand set; the second inequality is by the definition of an XOS function; and the last inequality is by the definition of $\OPT_{j,i}$, and because the probability that $j\in AV_i$ is at least $1$ minus the sum of probabilities that $j$ was sold up to block $\lceil i/n\rceil$ included. 

Since 
$\EE_{v\sim F^k}[\Welfare]   =  \EE_{v\sim F^k}[\Revenue] + \sum_{i \in [n k]} \EE_{v\sim F^k}[\Surplus_i], $
it is sufficient to lower bound for each $j\in M$, and $i\in [n]$, the sum of the coefficient of $\OPT_{j,i}$ in the revenue with the sum of coefficients of  $\OPT_{j,i}$ in the surplus, which is 
\begin{align*}
\sum_{t=1}^k q_{t,j} \left(1- \frac{1}{2^{k+1-t}}\right)  + \sum_{t=1}^k \frac{1-\sum_{t'=1}^{t}q_{j,t'}}{2^{k+1-t}}  & =  1-\frac{1}{2^k} +  \sum_{t=1}^k q_{t,j} \left(1- \frac{1}{2^{k+1-t}}\right)  - \sum_{t=1}^k \sum_{t'=t}^{k} \frac{q_{j,t}}{2^{k+1-t'}} \\ 
& =  1-\frac{1}{2^k}.  
\end{align*}
This concludes the proof.
\end{proof}

We next show a similar result with worse guarantees that uses static pricing.

\begin{theorem}\label{thm:xos-static}
For every $k\geq 1$, and every XOS combinatorial auction defined by $M$, $n$, and a product distribution $F=F_1\times\ldots\times F_n$,  there exists a static pricing scheme on $k$ copies such that 
$$ \EE_{v\sim F^k}\Big[\sum_{i\in [n k]} v_i(T_i)\Big] \geq  \left(1-\frac{1}{k+1}\right) \cdot \EE_{v\sim F}\Big[\max_{(T_1,\ldots,T_n)\in P(M) }\sum_{i\in [n]} v_i(T_i)\Big],$$
where $T_i$ is the demanded set purchased by each agent $i\in [n]$, and $P(M)$ is the set of all partitions of $M$ into $n$ sets. 
In particular, the $(1-\varepsilon)$-competition complexity of static prices for submodular and XOS combinatorial auctions is $O(1/\varepsilon)$.
\end{theorem}
\begin{proof}

We define price $p_j=(1-\frac{1}{k+1}) \sum_{i\in [n]}\OPT_{j,i}$, where $\OPT_{j,i}$ is defined as in the proof of Theorem~\ref{thm:xos-dynamic}.
Let $X_j$ be the indicator whether item $j$ was sold to one of the $nk$ agents, and
let $q_j = \Pr[X_j=1]$.
It holds that 
$$\EE_{v\sim F^k}[\Revenue] =\sum_{j\in M} q_j\cdot p_{j} = \sum_{j\in M}q_{j} \cdot \left(1-\frac{1}{k+1}\right)\cdot \sum_{i \in [n]} \OPT_{j,i} . $$
For every agent $i\in [nk]$, let's denote by $\AV_i$ the set of items that is available when $i$ arrives. For each type $r \in [n]$ let's denote by $T_{r}^*$ the set that agent $i$ of type $r$ receives in the welfare-maximizing allocation on profile $v \sim F$, and by $g_{r}^*$ the additive function that maximizes agent $r$'s value with respect to $v_i$ and $T_{r}^*$. Then, if agent $i$ is of type $r$,
\begin{align*} 
\EE_{v\sim F^k}[\Surplus_i] & =    \sum_{M' \subseteq M} \Pr[\AV_i = M'] \cdot \EE_{v_i\sim F_{r}} \Big[\max_{T \subseteq M'} (v_i(T) -p(T))\Big] \\ 
& \geq  \sum_{M' \subseteq M} \Pr[\AV_i = M'] \cdot \EE_{v\sim F} [ v_{r}(T^*_{r}\cap M') -p(T^*_{r} \cap M')] \\ 
& \geq   \sum_{M' \subseteq M} \Pr[\AV_i = M'] \cdot \EE_{v\sim F} \Big[ \sum_{j \in T^*_{r} \cap M'} (g_{r}^*(j) -p(j))\Big] \\
& \geq    \sum_{j\in M} \left(1-q_j \right)\cdot 
\frac{\OPT_{j,r}}{k+1 }. 
\end{align*}
Since 
$\EE_{v\sim F^k}[\Welfare]   =  \EE_{v\sim F^k}[\Revenue] + \sum_{i \in [n k]} \EE_{v\sim F^k}[\Surplus_i], $
it is sufficient to lower bound for each $j\in M$, and $i\in [n]$, the sum of the coefficient of $\OPT_{j,i}$ in the revenue with the sum of coefficients of  $\OPT_{j,i}$ in the surplus, which is 
$$ q_{j} \left(1- \frac{1}{k+1}\right)  + \sum_{t=1}^k \frac{(1-q_j)}{k+1}    =  1-\frac{1}{k+1}.  
$$
This concludes the proof.
\end{proof}

%% file: pointmasses.tex
\newpage
\section{Dealing with Point Masses}
\label{app:pointmass}
In this appendix, we discuss the adaptions of our results to the case where the distributions have point masses.
To deal with point masses, we can use standard techniques (e.g., \cite{arnosti2022tight}) by allowing our algorithms to select, whenever observing a value that is equal to the threshold $\tau$, with a fixed probability $p$, and when the value exceeds $\tau$, with probability one.
In particular, every threshold strategy can be described by two parameters, $\tau$ and $p$.

To generalize our proofs
to the case where there are point masses, we do the following:
Given distributions $F_1,\ldots,F_n$, and a quantile $g\in (0,1)$, 
let $\tau^*=\sup_{\tau\ge 0} \prod_{i=1}^n \Pr_{v_i
\sim F_i}[v_i \leq \tau] \leq g$ and $q= \prod_{i=1}^n \Pr_{v_i
\sim F_i}[v_i \leq \tau^*]$.
For every $i\in [n]$, let $q_i= \Pr_{v_i\sim F_i}[v_i \leq \tau^*]$, and let $p_i= \Pr_{v_i\sim F_i}[v_i=\tau^*]$.

Then, it holds that $\prod_{i=1}^n q_i = q \geq g \geq \prod_{i=1}^n (q_i-p_i)$, and both are strict inequalities, unless $p_i=0$ for all $i\in [n]$. If $p_i=0$ for all $i\in [n]$, then each  threshold should be interpreted as to select the first element that exceeds $\tau^*$.
Otherwise, let $p^*$ be the unique value satisfying $\prod_{i=1}^n (q_i-p^*\cdot p_i) = g$. 
Observe that $p^*$ must be in $[0,1]$.
Then, the algorithm should select the first element that exceeds $\tau^*$, and should accept the value $\tau^*$ with probability $p^*$. 
Note that $p^*$ is independent of the order of the distributions, and can be computed by only using distributional information.

We next show that if there are point masses, and if the algorithm must use a {\it deterministic} single threshold (i.e., that selects the first value that exceeds it), then the $(1-\varepsilon)$-competition complexity is in $\Theta(1/\varepsilon)$.
To prove this we first show that the $(1-\varepsilon)$-competition complexity is bounded from below by $\Omega(1/\varepsilon)$.
\begin{proposition}
\label{prop:static-mass}
For every $n\geq 2$, and for every $\varepsilon\in(0,1)$ there exists an instance with $n$ i.i.d. values that are distributed according to some distribution $F$ with point masses, such that for every $\tau$, and every $k  < \frac{  1}{60 \varepsilon}$, it holds that
\begin{equation} \label{eq:upper-static-mass}
    \EE_{v \sim F^{k\cdot n}}[\algt{\tau}(v)] < (1-\varepsilon) \cdot \EE_{v \sim F^n}[\max_{i \in [n]} v_i].
\end{equation}
\end{proposition}
\begin{proof}
It is sufficient to consider $\varepsilon< 1/20$, since otherwise $k=0$, and the claim holds trivially.
Consider the distribution $F$ such that $v_i = 1+ 10\cdot \text{Bernoulli}(\varepsilon/(3n))$ for every $i\in [n]$.
The RHS of Equation~\eqref{eq:upper-static-mass} satisfies that 
\begin{align*}
     (1-\varepsilon) \cdot \EE_{v \sim F^n}[\max_{i \in [n]} v_i] &= (1-\varepsilon)  \Big(1 + 10 \Big(1-\Big(1-\frac{\varepsilon}{3n}\Big)^n\Big)  \Big)\\
     &\geq  (1-\varepsilon)  \Big(1 + 10 \Big(1-e^{-\varepsilon/3}\Big) \Big)\geq 1+2\varepsilon, \nonumber
\end{align*}
where the last inequality holds for every $\varepsilon\leq 1/20$.
Now consider a static threshold algorithm with a threshold $\tau$.
We have two cases.\\

\noindent \textbf{Case 1:}
If $\tau< 1$, then since the algorithm will select the first value $v_1$, and the LHS of Equation~\eqref{eq:upper-static-mass} satisfies that 
\begin{equation*}
\EE_{v \sim F^{k\cdot n}}[\algt{\tau}(v)]  =  \EE_{v \sim F}[v] =   1+\frac{10\varepsilon}{3n} < 1+2\varepsilon,
\end{equation*}
where the inequality is since $n\geq 2$.
Thus, Equation~\eqref{eq:upper-static-mass} holds.\\

\noindent \textbf{Case 2:} 
If $\tau \geq  1$, then the algorithm will only accept the value $11$, and the the LHS of Equation~\eqref{eq:upper-static-mass} satisfies that 
\begin{align*}
\EE_{v \sim F^{k\cdot n}}[\algt{\tau}(v)]  & =   \textstyle\Pr_{v \sim F^{k\cdot n}}[\text{there exists } i \in [nk]: v_i >1]\cdot 11  \\ 
& =  11 \cdot \Big(1-\Big(1-\frac{\varepsilon}{3n}\Big)^{nk}\Big) \leq  1,
\end{align*}
where the last inequality holds for every $n\geq 2$, every $\varepsilon\leq 1/20$, and every $k \leq 1/(60\varepsilon)$.
Thus, \eqref{eq:upper-static-mass} holds, and this concludes the proof.
\end{proof}

We next show that there is deterministic single threshold algorithm that guarantees a $(1-\varepsilon)$-competition complexity of $O(1/\varepsilon)$.
\begin{proposition}
\label{prop:static-mass-upper}
For every $\varepsilon>0$ the
$(1-\varepsilon)$-competition complexity of the class of deterministic single threshold algorithms for the case where there are point masses is at most  $\lceil 1/\varepsilon \rceil$.
\end{proposition}
\begin{proof}
Given an instance $(F_1,\ldots,F_n)$, consider the threshold $\tau=(1-\varepsilon)\cdot\EE_{v\sim (F_1,\ldots,F_n)}[\max_{i\in [n]} v_i]$, and
let $p=\Pr_{v\sim (F_1,\ldots,F_n)}[\max_{i\in [n]} v_i \leq \tau]$.
Then, for $k=\lceil1/\varepsilon\rceil$, it holds that 

\begin{align*}
\EE_{v \sim (F_1, \dots, F_n)^k}[\algt{\tau}(v)] &  =  (1-p^k)\tau + \sum_{\ell=1}^k\sum_{i=1}^n \EE_{v \sim (F_1,\ldots,F_n)^k }[\max(v_i^{(\ell)}-\tau,0)]\cdot p^{\ell-1}\prod_{j<i}\Pr[v_j^{(\ell)} \leq \tau] \\ 
& \geq  (1-p^k)\tau + p^k\sum_{\ell=1}^k \sum_{i=1}^n \EE_{v_i\sim F_i}[\max(v_i-\tau,0)] \\ 
& \geq  (1-p^k)\tau + p^k \sum_{\ell=1}^k \EE_{v\sim (F_1,\ldots,F_n)}[\max(\max_{i\in [n]} v_i-\tau,0)]  \\
& \geq  (1-p^k)\tau + p^k\cdot \frac{1}{\varepsilon}  \cdot \EE_{v\sim (F_1,\ldots,F_n)}[\max_{i\in [n]} v_i-\tau]  \\
& =  (1-p^k)\tau + p^k \cdot \EE_{v\sim (F_1,\ldots,F_n)}[\max_{i\in [n]} v_i] \\  
& \geq  (1-\varepsilon)\cdot \EE_{v\sim (F_1,\ldots,F_n)}[\max_{i\in [n]} v_i],
\end{align*}
which concludes the proof.
\end{proof}

%% file: arrivalorder.tex
\section{Alternative Arrival Orders}
\label{sec:arrivalorder}

In this appendix, we study the competition complexity of the single-choice problem beyond the block model.
We first observe that our results continue to hold if the variables in each block arrive in arbitrary order. We then consider a version where variables can be moved across blocks, but in a limited way. Finally, we establish a lower bound for the fully adversarial case where the order can be arbitrary.

\paragraph{Displacement Model.} 
In our basic model --- the block model --- we assumed that variables arrive in the same order within each block. 
In the {\it displaced} block model, within each of the $k$ blocks the $n$ values sampled from the distributions $F_1,\ldots,F_n$ are presented in arbitrary order.

A first observation is that our proofs of the competition complexity results for block threshold algorithms and single threshold algorithms continue to hold in the displaced block model.

\begin{observation}\label{obs:displace}
 In the displaced block model, for every $\varepsilon>0$ the
 $(1-\varepsilon)$-competition complexity of the class of block threshold algorithms is $\Theta\left(\log\log(1/\varepsilon)\right)$.
 Furthermore, the $(1-\varepsilon)$-competition complexity of the class of single threshold algorithms is $\Theta\left(\log(1/\varepsilon)\right)$.
\end{observation}

Naturally, this motivates the question of studying the competition complexity in a setting where the values can be displaced across the different blocks.
In what follows, for a pair of natural numbers $(n,k)$ we denote by $S_{n}$ the set of permutations of $[n]$
and by $\sigma(F_1,\ldots,F_n)^k$ the instance that permutes $(F_1,\ldots,F_n)^k$ according to the permutation $\sigma\in S_{nk}$.

We introduce the {\it $\gamma$-displacement} model in which the adversary may change the order of $\gamma k$ copies of an instance $F = (F_1, \ldots, F_n)$, in a way such that if we split up the $\gamma n k$ variables into $\gamma$-blocks of size $\gamma n$, then within each $\gamma$-block every distribution appears at least once.
In what follows, we define a few objects to formally introduce this model.

Given a positive integer $\gamma$ and an integer $i$,
let $B_{\gamma,n}(i) =  
\{(i-1)\gamma n +j : j\in [\gamma n]\}$ be the $i$-th $\gamma$-block. 
Given an order $\sigma\in S_{\gamma nk}$,  and $j\in [n]$, let $I(\sigma,j) $ be the indices of $\sigma$ where there is a variable of type $j$.  Finally, let $D_{\gamma}(n,k)$ be the subset of permutations in $S_{\gamma nk}$ satisfying that, after permutation according to $\sigma$, each $\gamma$-block, contains at least one of each index in $[n]$, that is 
$$D_{\gamma}(n,k) = \Big\{\sigma \in S_{\gamma nk} : \text{for every } j\in [n]\text{ and } \text{every } i \in [ k], \ | I(\sigma,j) \cap B_{\gamma,n}( i)|\geq 1 \Big\}.$$

In the $\gamma$-displacement model, we compare $\min_{\sigma \in D_{\gamma}(n,k)}\EE_{v \sim \sigma F^{\gamma k}}[\algt{{\boldsymbol \tau}} (v)]$ to $\EE_{v\sim F}[\max_{i\in [n]} v_i]$, where $F=(F_1,\ldots,F_n)$.
We remark that the 1-displacement model corresponds to the displaced model discussed in Observation \ref{obs:displace}. In the following proposition, we show that the competition complexity in the $\gamma$-displacement model increases by at most a factor $\gamma$ in comparison to the $1$-displacement model.
\begin{proposition}\label{prop:displace}
In the $\gamma$-displacement model, for every $\varepsilon>0$ the $(1-\varepsilon)$-competition complexity of the class of block threshold algorithms is $O\left(\gamma \log\log(1/\varepsilon)\right)$.
\end{proposition}
\begin{proof}
To prove the result, we provide a reduction of the $\gamma$-displacement model to the block model. 
In the block model with $k$ copies, let ${\boldsymbol \tau} = (\tau_1,\ldots,\tau_k)$ be the thresholds corresponding to the quantiles $p_{\ell} =\Pr_{v \sim (F_1, \dots, F_n)}[\max_{j\in [n]}v_j \geq \tau_{\ell}]$ for $\ell \in [k]$.
Given some permutation $\sigma \in D_{\gamma}(n, k)$, 
define the thresholds ${\boldsymbol\tau'}=({\tau}_1',\ldots,{\tau}_k')$ to be such that $\Pr_{v\sim \sigma F^{\gamma k}}[\max_{j\in B_{\gamma,n}(\ell)}v_j \geq {\tau_\ell'}] = p_\ell$ for $\ell \in [k]$. 
Then, in the instance $\sigma(F_1,\ldots,F_n)^{\gamma k}$ with $\gamma k$ copies, consider the block thresholds ${\boldsymbol \tau'} =(\tau_1',\ldots,\tau_1',\tau_2',\ldots,\tau_2',\ldots,\tau_k',\ldots, \tau_k')$ where each threshold $\tau_\ell'$ for every $\ell \in [k]$ is repeated consecutively $\gamma$ times. 
Since $\sigma\in D_{\gamma}(n,k),$ for every $x\ge 0$ and every $\ell\in [k]$,
$$\textstyle\Pr_{v\sim \sigma F^{\gamma k}}[\max_{j\in B_{\gamma,n}(\ell)}v_j \geq x] \geq \Pr_{v\sim F}[\max_{j\in [n]} v_j \geq x],$$
that is, we have a stochastic dominance inequality.
We can then adjust the block model proof in Section \ref{sec:block} on $\gamma$-blocks instead of single blocks to get that $O\left(\gamma \log\log(1/\varepsilon)\right)$ copies provides a $(1-\varepsilon)$-competition complexity in the $\gamma$-displacement model. 
\end{proof}

We next show that an equivalent guarantee for the case of single-threshold algorithms cannot be achieved: namely, the $(1-\varepsilon)$-competition complexity of this type of algorithms in the $\gamma$-displacement model has to scale at least polynomially in $1/\varepsilon$.

\begin{proposition}\label{prop:displace-single-threshold}
For the $\gamma$-displacement model, there exists a $\gamma > 0$ such that the $(1-\varepsilon)$-competition complexity of any single-threshold algorithm is at least $1/(6\varepsilon^{1/3}).$
\end{proposition}

\begin{proof}
    Consider $\varepsilon > 0$ small enough,  such that $0 < \varepsilon < 0.076$. For ease of presentation,  suppose that $1/\varepsilon$ is an integer. Consider the following problem instance with $n=\frac{1}{\varepsilon}+2$ random variables of three different types:
    There is one random variable of type 1, with distribution $1+U[0,\varepsilon^{10}]$.
    There is one random variable of type $2$, whose distribution is $1+\varepsilon^{1/3} + U[0,\varepsilon^{10}]$ with probability $\varepsilon^{1/3}$ and $U[0,\varepsilon^{10}]$ otherwise. 
    There are $1/\varepsilon$ many random variables of type $3$, with distribution $U[0,\varepsilon^{10}]$. 
    So we basically have a deterministic one (type 1), a high value with low probability (type 2), and many zeros (type 3) --- plus some random noise.
    We denote the distributions of these random variables by $F_1,F_2,$ and $F_3$. Let $F = F_1 \times F_2 \times F_3 \times \ldots \times F_3$.

    The prophet can take the high value (the $\approx 1+\varepsilon^{1/3}$ value) from the second random variable if it realizes, or the guaranteed value of $\approx 1$ from the first random variable if it doesn't. It can thus achieve an expected value of at least 
    \begin{align*}
    \OPT \geq \varepsilon^{1/3} \cdot (1+\varepsilon^{1/3}) + (1- \varepsilon^{1/3}) \cdot 1 = 1 + \varepsilon^{2/3}. 
    \end{align*}

    Consider the $\gamma$-displacement model for $\gamma=3$ and $k$ copies.   We want to show that any single-threshold algorithm, in order to achieve a $(1-\varepsilon)$-approximation to the prophet, must have $k \geq 1/(6\varepsilon^{1/3})$. Assume towards contradiction that $k <1/(6\varepsilon^{1/3})$.

    Recall that in the $\gamma$-displacement model we seek a guarantee that applies for any possible grouping of the $\gamma n k$ random variables, into $k$ many $\gamma$-blocks of size $\gamma n$ each. We construct a hard instance as follows.
The first $\gamma$-block is: 
\[
(\underbrace{1,\ldots,1}_{2k+1},2,2,2,\underbrace{3,\ldots,3}_{\frac{3}{\varepsilon} +2-2k})
\]
All the remaining $\gamma$-blocks are:
\[
    (1, 2,2,2, \underbrace{3,\ldots,3}_{\frac{3}{\varepsilon}+2})
\]

Here $1, 2,$ and $3$ indicate that the respective random variable is of that type. So in the first $\gamma$-block we have $2k+1$ random variables of type 1, followed by $3$ random variables of type $2$, followed by $3/\varepsilon+2-2k$ random variables of type $3.$ Similarly, in the remaining $\gamma$-blocks we have one random variable of type $1$, followed by $3$ random variables of type $2$, followed by $3/\varepsilon+2$ random variables of type $3$. Note that both types of blocks consist of $\gamma n = 3/\varepsilon + 6$ random variables, and that variables of type $1$ and type $2$ occur for a total of $3k$ times while variables of type $3$ occur for a total of $3k/\varepsilon$ times. So this forms a valid instance.

Let's now analyze the performance of an arbitrary single-threshold algorithm $\ALG$. Let $T \geq 0$ be the algorithm's threshold, and let $q = \Pr_{X\sim F_1}[X \leq T]$. 
\[
\Pr[\ALG \geq 1] 
\leq 1- q^{3k}
(1-\varepsilon^{1/3})^{3k} \leq 1- q^{3k}(1-3k\varepsilon^{1/3}),
\]
where the last inequality uses that $(1-x)^{3k} \geq (1-3kx)$ for $x = \varepsilon^{1/3}$.

Next we show an upper bound on the probability that the algorithm stops on a value that is at least $1+\varepsilon^{1/3}$. For this the algorithm must skip over all the random variables of type $1$ in the first $\gamma$-block. Therefore,
\[
\Pr[\ALG \ge 1+\varepsilon^{1/3}] \leq 
q^{2k+1}\leq q^{2k},
\]
where we used that $q \leq 1$.
We thus obtain,
\begin{align*}
\EE[\ALG] 
&\leq \Pr[\ALG \geq 1] + \varepsilon^{1/3} \cdot \Pr[\ALG \geq 1+\varepsilon^{1/3}] + \varepsilon^{10} \notag\\
&= 1- q^{3k}
(1-3k\varepsilon^{1/3}) + \varepsilon^{1/3}  q^{2k} + \varepsilon^{10}. 
\end{align*}

We need that $\EE[\ALG] \geq (1-\varepsilon) \OPT$.  This gives us the following inequality 
\[
1- q^{3k}
(1-3k\varepsilon^{1/3}) + \varepsilon^{1/3}  q^{2k}+ \varepsilon^{10}
\geq (1-\varepsilon) (1+\varepsilon^{2/3}),
\]
which by rearrangement gives us that 

\[
\underbrace{\varepsilon+\varepsilon^{5/3} +\varepsilon^{10}}_{< \frac{\varepsilon^{2/3}}{2}}+
\underbrace{\varepsilon^{1/3} q^{2k}}_{\leq q^{1.5k} \varepsilon^{1/3}} + \underbrace{3k\varepsilon^{1/3}q^{3k}}_{\leq \frac{q^{3k}}{2}  }  \geq q^{3k} +\varepsilon^{2/3}  ,
\]
where the underlined inequalities hold since $\varepsilon<0.076$, $q \leq 1$, and 
$k < 1/(6\varepsilon^{1/3})$.

However, by the short multiplication formula\footnote{The short multiplication formula that we are using here is $(a-b)^2 = a^2 - 2ab + b^2$. Since $(a-b)^2 \geq 0$, this shows that $a^2 + b^2 \geq 2ab$. We apply this to $a = \sqrt{\varepsilon^{2/3}/2}$ and $b = \sqrt{q^{3k}/2}$.} we get that 
$$ q^{3k} +\varepsilon^{2/3} =  \frac{\varepsilon^{2/3}}{2} + \frac{q^{3k}}{2} + \left(\sqrt{\frac{\varepsilon^{2/3}}{2}}^2 + \sqrt{\frac{q^{3k}}{2}}^2\right) \geq \frac{\varepsilon^{2/3}}{2} + \frac{q^{3k}}{2} + q^{1.5k}\varepsilon^{1/3}.$$
This yields a contradiction, which implies that 
$k \geq 1/(6\varepsilon^{1/3})$.
\end{proof}

\paragraph{Fully Adversarial Model.} Beyond bounded displacement, a natural question is to ask what can be said in the case where the adversary is not restricted to choosing an order among a limited family of permutations, but can instead arrange the $nk$ distributions in any order. We refer to this as the {\it fully adversarial model}.  
It corresponds to the case where the instance of the online algorithm is given by $\sigma(F_1,\ldots,F_n)^k$, for an arbitrary $\sigma \in 
 S_{nk}$.
That is, each distribution is used $k$ times to sample values, but the order in which they are presented to the online algorithm is arbitrary.
We show that in this case the $(1-\varepsilon)$-competition complexity is lower bounded by $\Omega(1/\varepsilon)$.
\begin{proposition}\label{prop:adversarial}
For all $n\geq 2$, $\varepsilon \in (0,1)$ and every  $k<1/(4\varepsilon)$, there exist distributions $F_1,\ldots,F_n$ and a permutation $\sigma \in S_{nk}$ such that for every sequence of thresholds ${\boldsymbol \tau}$ we have
$$\EE_{v\sim \sigma(F_1,\ldots,F_n)^k}[\ALG_{\boldsymbol\tau}(v)]<(1-\varepsilon)\cdot\EE_{v\sim (F_1,\ldots,F_n)}[\textstyle \max_{i\in [n]}v_i].$$
\end{proposition}
\begin{proof}
First note that for $\varepsilon \in [1/2,1)$, the condition $k<1/(4\varepsilon)$ implies $k=0$. Hence we may assume $\varepsilon \in (0,1/2)$. Let $F_1 = \ldots = F_{n-1} = 1$ with probability $1$, and $F_n = 2$ with probability $2\varepsilon$ and $F_n=0$ otherwise. 
Let $\sigma$ be the permutation that sends the $k$ copies of $F_n$ to be the last $k$ distributions seen by the algorithm. Observe that a single threshold algorithm is optimal for this instance. 
Then, consider any single threshold $\tau\in (1,2]$. Then for every $\varepsilon \in (0,1/2)$ and every $k<1/(4\varepsilon)$, we have
$$\EE_{v\sim \sigma(F_1,\ldots,F_n)^k}[\ALG_{\tau}(v)] = (1-(1-2\varepsilon)^k)\cdot2 \leq 7/8.$$
Hence, the optimal threshold is at most 1, and yields $\EE_{v\sim \sigma(F_1,\ldots,F_n)^k}[\ALG_{\tau}(v)] =1.$
On the other hand, 
$$\EE_{v\sim (F_1,\ldots,F_n)}[\textstyle\max_{i\in [n]}v_i] = 2\varepsilon\cdot 2 + (1-2\varepsilon)\cdot 1 = 1+2\varepsilon.$$
This concludes the proof, as for every $\varepsilon \in (0,1/2)$, we have $(1-\varepsilon)(1+2\varepsilon)>1.$
\end{proof}

%% file: expectedruntime.tex
\section{Expected Number of Blocks}
\label{sec:expectedruntime}

In this appendix, we make a remark on the expected performance of the single threshold algorithm. Recall that its competition complexity in the block model is $\Theta(\log(1/\varepsilon))$ (Theorem \ref{thm:static}).

\begin{proposition}
There is a single threshold algorithm that, in expectation, terminates after two blocks and achieves an expected value of at least $\EE_{v\sim F}[\max_{i\in [n]} v_i]$.
\end{proposition}
\begin{proof}
We consider the same algorithm as in Lemma \ref{lem:static},
namely a single-threshold algorithm, with threshold $\tau^*$ satisfying 
$\prod_{i=1}^n \Pr_{v_i
\sim F_i}[v_i \leq \tau^*] = \frac{1}{2}.$
Equation \eqref{eq:ex_th} shows that the expected value achieved by this algorithm, conditional on stopping, is at least $\EE_{v\sim F}[\max_{i\in [n]} v_i]$.
We additionally observe that the expected number of blocks after which it stops is given by
$\sum_{i=1}^{\infty}i\left(1/2\right)^i = 2.$
\end{proof}

We remark that the preceding result is essentially tight: With $k$ copies, no algorithm can get a better approximation to $\EE_{v\sim F}[\max_{i\in [n]} v_i]$ than $k/2$. Thus, the expected number of blocks to achieve a $(1-\varepsilon)$-approximation, is at least $2-2\varepsilon$.

%% file: combinatorial.tex
\section{Matching Feasibility Constraints}
\label{app:matching}
In this appendix, we study the competition complexity for the vertex arrival model in bipartite graphs with one-sided arrival.
This problem falls in the more general XOS combinatorial auctions setting of Section \ref{sec:xos}.

More specifically, in this setting there is an underlying bipartite graph $G=(U,V,E)$, and the feasibility constraint is given by the set of matchings in $G$, that is, $\feas = \{S\subseteq E : \text{for all } i\in U\cup V, |\{e\in S : i\in e \}| \leq 1 \} $, and the valuation function is additive, i.e., $f(v,S) =\sum_{e \in S} v_e$.
The vertices of $V$ arrive online, one by one, and upon their arrival, their edges to all vertices in $U$ are revealed.
For every edge $e\in E$, the value $v_e$ is sampled according to a distribution $F_e$, and we denote $F= \bigtimes_{e\in E} F_e$.

Theorem \ref{thm:xos-dynamic} implies that the $(1-\varepsilon)$-competition complexity of block-consistent prices for this matching setting is $O(\log(1/\varepsilon))$.
In what follows, we show how to obtain this result by a different approach based on online contention resolution schemes. 
Our Algorithm \ref{alg:vertex} extends the one proposed by~\cite{ezra2022prophet} for vertex arrival model in bipartite graphs with one-sided arrival\footnote{Ezra et al.~\cite{ezra2022prophet} showed that the result holds also with respect to non-bipartite graphs.}.

\begin{algorithm}
    \caption{Online contention resolution scheme}
    \label{alg:vertex}
    
    Initialize $\ALG\leftarrow \emptyset$\;
    
    \For{$i\in \{1,\ldots,k\}$}{
        \For{$j\in\{1,\ldots,|V|\}$}{
            1. Observe $v_{j,u}^{(i)}$ for every $u\in U$\;
            2.  Sample $\tilde{v}_e\sim F_e$ for every $e\in E$ not incident to $j$\;
            3. Compute a maximum matching $\mu^{\star}$ in $G$ with edge-weights $w$ as follows:\\
            $w_{j,u}=v^{(i)}_{j,u}$ for every $u\in U$, and $w_e=\tilde{v}_e$ for every other edge\;
            
            4. \If{$u^*\in U$, the partner of $j$ in $\mu^{\star}$, exists and is available}{
                Match $j$ with $u^*$ with probability $1/(2 - \sum_{j'<j} x_{(j',u^*)})$, where $x_e = \Pr[e\in \arg\max_{S\in \feas}\sum_{e\in S} v_e ]$ for every $e\in E$\;
                Update $\ALG\leftarrow \ALG\cup \{(j,u^*)\}$\;
            }
        }
    }
    \Return $\ALG$\;
\end{algorithm}

\begin{theorem}\label{thm:matching}
For every $k\geq 1$, for every bipartite graph $G=(U,V,E)$, and every $F= \bigtimes_{e\in E} F_e$, Algorithm \ref{alg:vertex} always returns a matching in $G$, and it holds that
$$\EE_{v\sim F^k}[\textstyle\sum_{e\in \ALG(v)} v_e] \geq  \displaystyle\left(1-\frac{1}{2^k}\right) \cdot \EE_{v\sim F}[\textstyle\max_{S\in \feas }\sum_{e\in S} v_e],$$
where $\calF$ is the set of matchings in $G$. 
In particular, the $(1-\varepsilon)$-competition complexity of Algorithm~\ref{alg:vertex} for the online matching problem with one-sided vertex arrival is $O(\log(1/\varepsilon))$.
\end{theorem}

\begin{proof}
First, observe that the algorithm is well defined since $2 -  \sum_{j'<j} x_{(j',u^*)} \geq 2 -  \sum_{j'} x_{(j',u^*)} \geq 2-1 = 1$.
We assume without loss of generality that for every $u\in U$, it holds that $\sum_{j\in V} x_{(j,u)} =1 $. This assumption can be made by adding $|U|$ auxiliary vertices that have edges to all edges in $U$, and all their edges always have a value of zero.
It holds that  
\begin{equation} \label{eq:bound-opt}
    \EE_{v\sim F}\Big[\max_{S\in \feas }\sum_{e\in S} v_e\Big]  = \sum_{e\in E} \Pr_{v\sim F} \Big[e\in \arg\max_{S\in \feas }\sum_{e\in S}v_e\Big] \cdot \EE_{v\sim F}\Big[v_e \mid e\in \arg\max_{S\in \feas }\sum_{e\in S}v_e\Big].
\end{equation}
On the other hand,  for every $i\in [k]$, and every $e\in E$, it holds that 
\begin{equation} \label{eq:bound-algv}
    \EE_{v\sim F^k}[v_e^{(i)} \mid  v_e^{(i)} \in \ALG(v)]  =  \EE_{v\sim F}\Big[v_e \mid e\in \arg\max_{S\in \feas }\sum_{e\in S}v_e\Big],
\end{equation}
since $(v^{(i)}_{(j,u)})_{u\in U}$, $(\tilde{v}_{(j',u)})_{j' \in V \setminus\{j\}, u\in U}$ are distributed the same as $v\sim F $.
We refer to time $(j,i)$ to the arrival of the $j$-th node in block $i$.
In what follows, we prove by induction that for every $u\in U$ \begin{equation}
    \label{eq:induction}
\Pr[u \text{ is available at time } (j,i)] = \frac{2 -  \sum_{j'<j} x_{(j',u)}}{2^i}.\end{equation}
The base case is when $i=1$, and $j$ is the first vertex to arrive, and then it holds that $(2 -  \sum_{j'<j} x_{(j',u)})/2^i =(2-0)/2 =1$.
Assume by induction that it is true if $j$ is not the first vertex to arrive at block $i$, then for $i$ and $j-1$, and else $j$ is the first to arrive at block $i>1$, then for $i-1$ and $j_{\text{last}}\in V$, where $j_{\text{last}}$ is the last vertex to arrive in block $i-1$.
If $j$ is not the first vertex of block $i$, then if we denote by $j_{\text{pre}}$ the vertex that arrives before $j$ in block $i$, then
\begin{align*}
    &\Pr[u \text{ is available at time } (j,i)] \\ 
    &= \Pr[u \text{ is available at time } (j_{\text{pre}},i)] \left(1- \Pr[u \text{ is the partner of } j_{\text{pre}} ] \cdot  \frac{1}{2 -  \sum_{j'<j_{\text{pre}}} x_{(j',u)}}\right) \\ &= \frac{2 -  \sum_{j'<j_{\text{pre}}} x_{(j',u)}}{2^i} \left(1- x_{(j_{\text{pre}},u)} \cdot  \frac{1}{2 -  \sum_{j'<j_{\text{pre}}} x_{(j',u)}}\right) 
    \\
    &= \frac{2 -  \sum_{j'<j_{\text{pre}}} x_{(j',u)}}{2^i} \left(  \frac{2-\sum_{j'<j} x_{(j',u)}}{2 -  \sum_{j'<j_{\text{pre}}} x_{(j',u)}}\right)  = \frac{2 -  \sum_{j'<j} x_{(j',u)}}{2^i}.
\end{align*}
Else, if $j$ is the first vertex of block $i$, then
\begin{align*}
    &\Pr[u \text{ is available at time } (j,i)] \\ 
    &=\Pr[u \text{ is available at time } (j_{\text{last}},i-1)] \left(1- \Pr[u \text{ is the partner of } j_{\text{last}} ] \cdot   \frac{1}{2 -  \sum_{j'<j_{\text{last}}} x_{(j',u)}}\right) \\ 
    &= \frac{2 -  \sum_{j'<j_{\text{last}}} x_{(j',u)}}{2^{i-1}} \left(1- x_{(j_{\text{last}},u)} \cdot  \frac{1}{2 -  \sum_{j'<j_{\text{last}}} x_{(j',u)}}\right) 
    \\ 
    &= \frac{2 -  \sum_{j'<j_{\text{last}}} x_{(j',u)}}{2^{i-1}} \left(  \frac{2-\sum_{j'\leq j_{\text{last}}} x_{(j',u)}}{2 -  \sum_{j'<j_{\text{last}}} x_{(j',u)}}\right)  = \frac{2 -  \sum_{j'<j} x_{(j',u)}}{2^i},
\end{align*}
where the last equality is since $\sum_{j'<j} x_{(j',u)}=0 $, and $\sum_{j'\leq j_{\text{last}}} x_{(j',u)} = 1$, which concludes the proof of the induction.
Finally, for every $e\in E$, we have 
\begin{align} 
\sum_{i\in [k]} \Pr[v_e^{(i)} \in \ALG(v)]  &=  x_e \cdot \frac{1}{2 -  \sum_{j'<j} x_{(j',u^*)}} \cdot \sum_{i\in [k]} \Pr[u^* \text{ is available at time } (j,i)] \nonumber\\ 
& =   x_e \cdot\frac{1}{2 -  \sum_{j'<j} x_{(j',u^*)}} \cdot \sum_{i\in [k]} \frac{2 -  \sum_{j'<j} x_{(j',u^*)}}{2^i}  =x_e\cdot \left(1- \frac{1}{2^k}\right),\label{eq:alg-prob}
\end{align} 
where the second equality holds by Equation~\eqref{eq:induction}.
The theorem then holds by combining \eqref{eq:bound-opt}, \eqref{eq:bound-algv}, and \eqref{eq:alg-prob}, together with $ \Pr_{v\sim F} [e\in \arg\max_{S\in \feas }\sum_{e\in S}v_e] =x_e$.
\end{proof}